\documentclass[aps,groupedaddress]{revtex4}
\usepackage{epsfig}
\usepackage{graphicx}
\usepackage[T1]{fontenc}
\usepackage{ae}
\usepackage{ulem}
\usepackage{color}
\usepackage[latin1]{inputenc}
\usepackage{amssymb,amsbsy,amsmath}
\usepackage{bbm}

\newtheorem{lemma}{Lemma}
\newtheorem{defi}{Definition}
\newtheorem{prop}{Proposition}

\newtheorem{theorem}{Theorem}
\newtheorem{cor}{Corollary}

\newtheorem{ass}{Assertion}
\newtheorem{axiom}{Axiom}

\input ulem.sty

\begin{document}

%%%%%%%%%%%%%%%%%%%%%%%%%%%%%%%%%%%%%%%%%%%%%%%%%%%%%%%%%%%%%%%%%%%%%%%%%%%%%
%\begin{center}
 % \today
%\end{center}

\title[QOCS]
      {Quantum orientation, Noether structure, composition of systems  and  operations}
\author{Heinz-J\"urgen Schmidt$^1$}
\address{$^1$  Universit\"at Osnabr\"uck,
Fachbereich Mathematik, Informatik und Physik,
 D - 49069 Osnabr\"uck, Germany}

%\tableofcontents

\begin{abstract}
In this paper we argue that, in addition to the statistical structure of quantum theory,
another structure, referred to here as the ``Noether structure,"
is necessary to describe the composition of systems and to define completely positive operations.
A Noether structure reflects the dual role of Hermitian operators as
observables on the one hand and as generators of symmetry transformations on the other.
This idea has been expressed in a similar form in the works of Alfsen and Shultz,
who investigated the conditions under which the Jordan product
can be extended to an associative product of operator algebras.
Our investigations into the Noether structure and the composition of systems are limited to the
finite-dimensional case and establish a connection to completely positive operations.
In the case of pure operations, the latter can be characterized as orientation-preserving maps.
\end{abstract}

\maketitle
%%%%%%%%%%%%%%%%%%%%%%%%%%%%%%%%%%%%%%%%%%%%%%%%%%%%%%%%%%%%%%%%%%%%%%%%%%%%%%%%%%%%%%%%%%%%%%%%%%%%%%%%%%%%%%%%%%%%%%%%%%%%%%%
\section{Introduction}\label{sec:IN}
%%%%%%%%%%%%%%%%%%%%%%%%%%%%%%%%%%%%%%%%%%%%%%%%%%%%%%%%%%%%%%%%%%%%%%%%%%%%%%%%%%%%%%%%%%%%%%%%%%%%%%%%%%%%%%%%%%%%%%%%%%%%%%%

The mathematical formalism commonly used to formulate quantum theory (QT) only has a partial physical interpretation.
For example, a central concept in the mathematical formulation of QT is a complex Hilbert space ${\mathcal H}$,
but only a subset of the Hilbert space vectors, namely of those with norm $1$, has a physical interpretation in terms of pure states,
not to mention the question of whether a phase factor $\exp({\sf i} \alpha)$ of a normalized vector has any physical meaning.
Similarly, the complex scalar product $\langle \phi \left| \right. \psi \rangle$ between two vectors $\phi, \psi\in{\mathcal H}$
with $\| \phi \| =\|\psi\|=1$, which is sometimes referred to as an ``amplitude,"
has no direct physical meaning, in contrast to the transition probability
$\left|\langle \phi \left| \right. \psi \rangle\right|^2$.
Consequently, various approaches were pursued to develop a step-by-step formulation of
QT that does not suffer from this drawback. There was considerable activity in this area,
particularly in the 1960s and 1970s; see, for example, \cite{V68,J68,P76,D79}.
These approaches differ in their choice of basic concepts and laws that are accepted as ``physical"
and in their scope.   An example of a particularly detailed and extensive approach is the
axiomatic foundation of quantum mechanics by G.~Ludwig \cite{L70,L83,L85a,L85b,L87}.

According to Ludwig, the core of QT is, what we will call the ``statistical structure"
$\langle K,L,\mu\rangle$ with the physical interpretation of $K$ as the set of ``ensembles" (pure or mixed states),
of $L$ as the set of ``effects" (generalized yes-no-experiments)
and of $\mu(W,F)$ as the probability with which the effect $F\in L$ occurs in the ensemble $W\in K$.
In the simplest case, this structure is realized by:
\begin{eqnarray}
\label{realK}
  K&=& K({\mathcal H}):=\mbox{ set of statistical operators
acting on the Hilbert space } {\mathcal H},\\
\label{realL}
 L&=&L({\mathcal H}):= \mbox{ set of Hermitian operators with eigenvalues } 0\le \lambda \le 1
\mbox{ acting on the Hilbert space } {\mathcal H},\\
\label{realmu}
\mu &=& \mu({\mathcal H}) \quad \quad \mbox{ given by }
\mu(W,F)=\mbox{Tr}(W\,F) \mbox{ for } W\in K({\mathcal H}) \mbox{ and } F\in L({\mathcal H})
\;,
\end{eqnarray}
which will be called the ``standard realization".

The detailed physical interpretation of the statistical structure
is provided by a pre-theory of
preparation and measurement procedures, which we will not discuss here.
In Ludwig's axiomatic foundation of QT, one seeks for axioms for $\langle K,L,\mu\rangle$
that enable the standard realization. Ideally, these axioms have a direct physical interpretation
as ``fundamental principles of measurement" (Haupts\"atze des Messens).
Opinions may differ as to whether Ludwig's approach has been entirely successful in this regard.
For our purposes, we can replace the relevant axioms with a single one:
\begin{axiom}\label{AX}
There exists a complex, separable Hilbert space ${\mathcal H}$ and bijections $f: K \to K({\mathcal H}),\; g: L\to L({\mathcal H})$
such that\\
$\mu(W,F)= \mbox{Tr}(f(W)\,g(F))$.
\end{axiom}
This allows us to leave aside the question of the physical interpretation of the axioms for the purpose of this paper.
Note that by Axiom \ref{AX} the sets $K$ and $L$ are organized as convex sets embedded into the linear
(Banach) spaces $f(K)\subset B({\mathcal H})$ of Hermitean trace class operators, resp.,
$g(L)\subset B({\mathcal H})'$ of bounded Hermitean operators, the latter viewed as the dual space
of $B({\mathcal H})$ by means of the bilinear form $(X,Y)\mapsto \mbox{Tr}(XY)$.

However, it is obvious that the statistical structure of QT is insufficient
to completely reconstruct the representation of QT as we know it from textbooks.
For example, the formulation of Schr\"odinger's equation requires an extension
of the statistical vocabulary in order to represent time evolution through transformations of ensembles
(Schr\"odinger picture) or, alternatively, of effects (Heisenberg picture).
This was also Ludwig's view.
The statistical structure is not the end point
of the axiomatic foundation of QT, but rather the starting point for further
developing the physical interpretation of QT.
The latter example can also be seen as a hint to the problem of identifying
certain observables with the generators of one-parameter subgroups of automorphisms in QT.
We will come back to this problem below.

Another important concept in QT is the composition of systems,
which, in the standard QT formalism, is described by the tensor product of the Hilbert spaces of the subsystems.
In Section \ref{sec:CS}, we will show that the composition of subsystems
cannot be uniquely described in terms of the statistical structure of QT.
Closely related to this is another fundamental concept:
the completely positive state transformations (see Section \ref{sec:CPT}),
which are used, for example, to describe time evolution in open quantum systems, see, e.~g., \cite{S26}, or
transformations resulting from measurements that generalize the von Neumann projection postulate.
Here, too, the statistical structure of QT seems too limited to adequately define the term.

This raises the question of what additional structures would be capable of doing so.
An obvious requirement for such a structure would be that it reduces the freedom of
choice of Hilbert space to unitary isomorphisms. This would immediately invalidate
the above counterexample in Section \ref{sec:CS} regarding the non-uniqueness of the composition of systems,
as well as all similar counterexamples. A related question has been dealt with
in detail in the works of Alfsen, Shultz et al \cite{AHS80,AS98a,AS98b,AS12a,AS12b}.
These authors have investigated
the conditions under which Jordan algebras can be extended to associative $C^\ast$-algebras
or von Neumann algebras. The central condition for this has been found to be
``orientability," and in this case, the choice of the associative product
corresponds exactly to the choice between two ``orientations".
Precursors to the concept of quantum orientation can be found in the work of Alain Connes, see \cite{C74}.

A structure equivalent to quantum orientation was introduced by Alfsen and Shultz
as ``dynamical correspondence"; see the section \ref{sec:NSI}.
It denotes a linear mapping  -  which is essentially an isomorphism  - from the space of observables
to the Lie algebra of the automorphism group of a quantum system.
For example, the energy observable $H$ is also a generator of the time evolution
according to the Schr\"odinger equation. Thus, the ``dynamical correspondence" or
``Noether structure," as it is called in this paper, lends itself even more readily
to a physical interpretation than  the geometric concept of``quantum orientation". 
In Section \ref{sec:NSII},
we present a version of the Noether structure that also includes a dimensional
conversion factor, which we identify with $\hbar^{-1}$.

As mentioned earlier, Alfsen and Shultz use the Noether structure, or equivalently,
the quantum orientation, to define an associative product on the space of all
linear operators over a Hilbert space (which is always finite-dimensional in this paper).
This then allows for the straightforward definition of composite systems or completely positive operations.
However, the question arises as to whether the latter cannot be defined directly in terms
of the preservation of orientation; see Section \ref{sec:CPT}.
At least for the subset of ``pure" operations, this is indeed possible;
 see Theorem \ref{Tpure} in subsection \ref{sec:CPT2}.
These are operations that map pure states, up to a factor, back to pure states.
The proof is based on the characterization of pure operations given by
Davies in 1969 \cite{D69}; see Theorem \ref{TPO} in subsection \ref{sec:CPT2}.
It is quite complex and has therefore been moved in part to the appendix.
Finally, we summarize the results in Section \ref{sec:SUM}.

%%%%%%%%%%%%%%%%%%%%%%%%%%%%%%%%%%%%%%%%%%%%%%%%%%%%%%%%%%%%%%%%%%%%%%%%%%%%%%%%%%%%%%%%%%%%%%%%%%%%%%%%%%%%%%%%%%%%%%%%%%%%%%%
\section{Composition of systems}\label{sec:CS}
%%%%%%%%%%%%%%%%%%%%%%%%%%%%%%%%%%%%%%%%%%%%%%%%%%%%%%%%%%%%%%%%%%%%%%%%%%%%%%%%%%%%%%%%%%%%%%%%%%%%%%%%%%%%%%%%%%%%%%%%%%%%%%%

In the following, we will limit ourselves to systems that are described by a
finite-dimensional Hilbert space ${\mathcal H}$ with $\mbox{dim}({\mathcal H})=N$.
Furthermore, we will not use any super selection rules,
and hence, in the language of von Neumann algebras, only consider factors $I_N$
which is already implicitly assumed by the form of the realization
(\ref{realK}), (\ref{realL}) of statistical structures of QT.
This is a pragmatic decision to separate conceptual issues from mathematical complications.
Hence the statistical structure $\langle K,L,\mu\rangle$ is determined by the dimension $N$
of the Hilbert space ${\mathcal H}$, up to isomorphism (more on this later).
When composing systems, we limit ourselves to different systems and thus
exclude the theory of bosons and fermions, not because it is unimportant,
but because it would only complicate matters and distract from the fundamental problem.
Physical examples of composition of
finite-dimensional systems would be spin systems, each described by Hilbert spaces
of dimension $N_1=2s_1+1$ and $N_2=2 s_2+1$, where (anti-)symmetrization plays no role.

Throughout this paper  $B({\mathcal H})$ denotes the linear space of Hermitean operators
$X:{\mathcal H}\to {\mathcal H}$,   $B_0({\mathcal H})\subset B({\mathcal H})$
the subspace of Hermitean operators
with vanishing trace, ${\mathcal A}({\mathcal H})$ the space of all linear operators and
${\mathcal U}({\mathcal H})$ the $N^2$-dimensional Lie group of unitary operators.

We will attempt to formulate a ``conservative" extension of the statistical structure
for describing composite systems, which is, however, doomed to failure.
Given are two quantum systems described by Hilbert spaces ${\mathcal H}_1$ and ${\mathcal H}_2$
with dimensions $N_1$ and $N_2$. According to the usual rules of QT, the composite system
is described by the tensor product Hilbert space
${\mathcal H}_{12} = {\mathcal H}_1 \otimes {\mathcal H}_2$ with dimension $N_{12}=N_1 N_2$.
In statistical language, the composite system is a structure
$\langle K_{12}, L_{12}, \mu_{12}\rangle$.
Of course, a connection must still be established with the structures $\langle K_1, L_1, \mu_1\rangle$
and $\langle K_2, L_2, \mu_2\rangle$ belonging to the individual systems.
This can be done by injective, affine mappings
\begin{equation}\label{pq}
p: K_1  \times K_2  \to K_{12} \quad \mbox{ and } q: L_1 \times L_2 \to L_ {12}
\;,
\end{equation}
that satisfy the condition
\begin{equation}\label{condpq}
 \mu_{12}(p(W_1,W_2),q(F_1,F_2))=\mu_1(W_1, F_1)\, \mu_2(W_2, F_2)
 \quad \mbox{ for all } W_1\in K_1,\,W_2\in K_2, \,F_1 \in L_1,\, F_2\in L_2
\end{equation}
This is an axiomatization of the usual composition
$W_{12}= W_1 \otimes W_2$ of (mixed) states to product states
and of effects $F_ {12} = F_1 \otimes F_2$ to product effects,
whereby condition (\ref{condpq}) follows from
$\mbox{Tr}((W_1 \otimes W_2 )( F_1 \otimes F_2))= \mbox{Tr}(W_1 F_1) \mbox{Tr}(W_2 F_2)$.
Since the composite system $\langle K_{12}, L_{12}, \mu_{12}\rangle$ is assumed to satisfy the Axiom \ref{AX},
there exists a Hilbert space $\tilde{\mathcal H}_{12}$ through which $\langle K_{12}, L_{12}, \mu_{12}\rangle$
is realized in the usual way, which does not necessarily coincide {\it a priori}
with the tensor product space ${\mathcal H}_{12} = {\mathcal H}_1 \otimes {\mathcal H}_2$ considered above.
The question is whether the structure
$\langle K_{12}, L_{12}, \mu_{12}\rangle$  can be defined in terms of $\langle K_1, L_1, \mu_1\rangle$,
$\langle K_2, L_2, \mu_2\rangle$, $p$, $q$,  and any additional requirements.

Let us clarify this question a little further. To do so, we recall that in Ludwig's axiomatic foundation of QT,
the sets $K$ and $L$ are each embedded into Banach spaces $K\subset B$ and $L\subset B'$ (dual space)
such that $\mu:K\times L \to [0,1]$ coincides with the restriction of the canonical bilinear form on $B\times B'$.
In the standard realization over a Hilbert space ${\mathcal H}$, this embedding assumes the form
$ K({\mathcal H}) \subset B({\mathcal H})$ and $L({\mathcal H}) \subset B({\mathcal H})'$,  respectively.
In the finite-dimensional case,
$B({\mathcal H})$ and $B({\mathcal H})'$ can be identified by considering the Euclidean scalar product
$(X,Y) \mapsto \mbox{Tr}(X Y)$. We will do this in the following. In this case, the cone
$B^+({\mathcal H})$ of positive semi-definite Hermitian operators is ``self-dual,"
i.~e., $X\ge 0$ iff  $ \mbox{Tr}( X Y) \ge 0$ for all $Y\ge 0$.

When composing systems in standard quantum mechanics,
the two spaces $B({\mathcal H}_1 \otimes {\mathcal H}_2)$ and $B({\mathcal H}_1) \otimes B({\mathcal H}_2)$
can be identified. It therefore makes sense to also consider the sets $K_{12}$ and $L_{12}$ as subsets of
$B_1 \otimes B_2$ and to define the two mappings
\begin{eqnarray}
\label{preal}
 p&:& K_1 \times K_2  \to K_{12} \subset B_1\otimes B_2\\
 \label{qreal}
 q&:& L_1 \times L_2  \to L_{12 }\subset B_1 \otimes B_2
\end{eqnarray}
considered in (\ref{pq}) as tensor product mappings.
Then the situation is such that $K_{12}$ and $L_{12}$ are two initially
unknown subsets of $B_1\otimes B_2$ with the following properties:
\begin{eqnarray}
\label{i}
 && K_1 \otimes K_2 \subset K_{12}\subset B_1\otimes B_2 \;, \\
 \label{ii}
&&  L_1 \otimes L_2 \subset L_{12}\subset B_1\otimes B_2\;,\\
 \label{iii}
&& \langle K_{12}, L_{12}, \mu_{12}\rangle \quad \mbox{ satisfies the axiom \ref{AX}}
\;,
\end{eqnarray}
where $\mu_{12}$ is the Euclidean scalar product on  $B_1 \otimes B_2$.
The question then arises as to whether the sets $K_{12}$ and  $L_{12}$
are uniquely determined by the requirements (\ref{i}),  (\ref{ii}), and  (\ref{iii}).

In preparation to its answer, some remarks on the uniqueness of the Hilbert space ${\mathcal H}$ for the realization of a
statistical structure $\langle  K, L, \mu\rangle$ are in order.
The axiom \ref{AX} merely guarantees the existence of ${\mathcal H}$ .
With regard to the degree of uniqueness, the following can been proven,
see theorem 5.23 in \cite{L83} and \cite{MP01}, cp.~also Wigner's theorem, \cite{W59}.
%%%%%%%%%%%%%%%%%%%%%%%%%%%%%%%%%%%%%%%%%%%%%%%%%%%%%%%%%%%%%%%%%%%%%%%%%%%%%%%%%%%%%%%%%%%%%%%%%%%%%%%%%%%%%%
\begin{prop}\label{PU}
If there are two Hilbert spaces ${\mathcal H}_a$  and ${\mathcal H}_b$ that satisfy the axiom \ref{AX},
then there exists a unitary or antiunitary isomorphism $U: {\mathcal H}_a \to {\mathcal H}_b$.
\end{prop}
%%%%%%%%%%%%%%%%%%%%%%%%%%%%%%%%%%%%%%%%%%%%%%%%%%%%%%%%%%%%%%%%%%%%%%%%%%%%%%%%%%%%%%%%%%%%%%%%%%%%%%%%%%%%%

In this case, we obtain canonical bijections $G_K: K({\mathcal H}_a) \to K({\mathcal H}_b)$ and
$G_L: L({\mathcal H}_a) \to L({\mathcal H}_b)$, which are defined by
$G_K(W)= U W U^*$ and $G_L(F)= U F U^*$ and satisfy $\mbox{Tr}(G_K(W) G_L(F))= \mbox{Tr}(U W U^*U F U^*) =\mbox{Tr}(WF)$.
This immediately implies that the uniqueness of the Hilbert space is, at least, only guaranteed up to unitary or anti-unitary isomorphisms.
The less trivial part of the proof consists in showing that the unitary or anti-unitary
isomorphisms already exhaust all possible choices of the Hilbert space.

We will now show that the statistical structure $\langle K_{12}, L_{12}, \mu_{12}\rangle$
of a composite system is {\it not} uniquely determined by conditions  (\ref{i}),  (\ref{ii}), and  (\ref{iii}).
This finding is consistent with the results of \cite{AD78} and \cite{P04},
although these authors use different theoretical frameworks to describe the composition of quantum systems.

Since it suffices to find a single counterexample, we can restrict ourselves
to the simplest case $\dim {\mathcal H}_1=\dim {\mathcal H}_2 =2$.
Due to the freedom in the choice of the Hilbert space, we first consider
the case ${\mathcal H}_1={\mathcal H}_2= {\mathbb C}^2$.
The qubit ensembles $W\in K({\mathbb C}^2)$ can then be represented by Hermitian, positively semi-definite
$2\times 2$-matrices with trace $1$; analogously, the effects $F\in L({\mathbb C}^2)$
can be represented by Hermitian $2\times 2$-matrices with eigenvalues lying in the closed interval $[0,1]$.
The space $B({\mathcal H}_1 \otimes {\mathcal H}_2)$ can therefore be identified with the space of
Hermitian $4\times 4$-matrices. The bilinear embeddings
$p: K_1 \times K_2 \to K_{12} \subset  B_1 \otimes B_2$ and
$q: L_1 \times L_2 \to L_{12} \subset  B_1 \otimes B_2$ are thus uniquely determined and have the following form:
\begin{equation}\label{p12}
  \left( \left(
  \begin{array}{cc}
    \lambda & a \\
    \overline{a} & 1-\lambda
  \end{array}
  \right)
  \;,\;
  \left(
  \begin{array}{cc}
    \mu & b \\
    \overline{b} & 1-\mu
  \end{array}
  \right)\right)
 \stackrel{p}{\longrightarrow}
 \left(
 \begin{array}{cc|cc}
   \lambda \mu & \lambda b & a \mu & a b \\
   \lambda \overline{b} & \lambda (1-\mu) & a \overline{b} & a(1-\mu) \\
   \hline
   \overline{a} \mu&\overline{a}  b& (1-\lambda)\mu & (1-\lambda) b\\
   \overline{a}\overline{b} & \overline{a} (1-\mu) &  (1-\lambda) b &  (1-\lambda) (1-\mu)
 \end{array}
 \right)
 \;,
\end{equation}
where $|a|^2\le \lambda(1-\lambda),\;|b|^2\le \mu(1-\mu)$, and analogously for $q$.
Note that in the r.~h.~s.~of (\ref{p12}) the local blocks are proportional to the second factor
of the l.~h.~s.~such that the proportionality constants are given by the first factor.

The set $K_{12}$ is given by the ``positive" Hermitean $4\times 4$-matrices $X$ with unit trace.
Thus the obvious choice is to define ``positive" in the usual sense of positively semi-definite,
in symbols, $X\ge 0$, and to consider the corresponding statistical structure
$\langle K_{12}, L_{12}, \mu_{12}\rangle$ in $4$ dimensions.
The product states shown in (\ref{p12}) are also ``positive" in this sense and hence (\ref{i})
is satisfied, analogously for (\ref{ii}).

However, we may define ``positive" Hermitean $4\times 4$-matrices in an alternative sense.
To motivate this alternative definition we recall that the second factor of the tensor product
of two-dimensional Hilbert spaces may be chosen as ${\mathbb C}^{2\ast}$, the dual space of ${\mathbb C}^{2}$.
In fact, there exists an anti-unitary isomorphism $U: {\mathbb C}^{2} \to {\mathbb C}^{2\ast}$,
mapping ``kets" to ``bras". In components, $U$ assumes the form of the anti-linear map
$U(c_1,c_2)= (\overline{c_1},\overline{c_2})$. For Hermitean $2\times 2$-matrices $X$ the transformation
$X\mapsto U X U^\ast$ is equivalent to the transposition, denoted by $\tau$.
The map ${\mathbbm 1}\otimes \tau: B_1\otimes B_2 \to B_1\otimes B_2$
performs transpositions of the local blocks. We will define ``positive" Hermitean $4\times 4$-matrices $X$
by the condition
\begin{equation}\label{ge}
 X\succeq 0 \Leftrightarrow ({\mathbbm 1}\otimes \tau) (X)\ge 0
 \;,
\end{equation}
where the latter symbol $\ge 0$, as before, means positively semi-definite.
The alternative definition $\succeq 0 $ leads to an alternative statistical structure
in $4$ dimensions, denoted bx $\langle \widetilde{K_{12}}, \widetilde{L_{12}}, \widetilde{\mu_{12}}\rangle$,
where $\widetilde{\mu_{12}}$ operates in the same manner as $\mu_{12}$ by $(X,Y)\mapsto \mbox{Tr}(X Y)$
and differs only in its domain of definition. We have to show that
$\langle \widetilde{K_{12}}, \widetilde{L_{12}}, \widetilde{\mu_{12}}\rangle$ also
satisfies the conditions (\ref{i}),  (\ref{ii}), and  (\ref{iii}) and that it differs from
$\langle K_{12}, L_{12}, \mu_{12}\rangle$.

The latter follows since ${\mathbbm 1}\otimes \tau$ does not map
$B^+({\mathbb C}^{2} \otimes {\mathbb C}^{2})$ into itself. For example,
\begin{equation}\label{counter}
\left(
\begin{array}{cc|cc}
  1& 0 & 0 & 1 \\
  0 & 0 & 0 & 0 \\
  \hline
   0 & 0 & 0 & 0 \\
   1& 0 & 0 & 1 \\
\end{array}
 \right)
 \stackrel{{\mathbbm 1}\otimes \tau}{\longrightarrow}
 \left(
\begin{array}{cc|cc}
  1& 0 & 0 & 0 \\
  0 & 0 & 1 & 0 \\
  \hline
   0 & 1 & 0 & 0 \\
   0& 0 & 0 & 1 \\
\end{array}
 \right)
 \;.
\end{equation}
The former matrix has the eigenvalues $2,0,0,0$, and the latter one the eigenvalues $1,1,1,-1$.
This example is identical to the usual one showing that $\tau$ is 
monotone but not completely positive, see \cite{NC00}.

Concerning the conditions  (\ref{i}) and  (\ref{ii}) one can argue that
${\mathbbm 1}\otimes \tau$ maps product states, i.~e., states of the form (\ref{p12})
onto product states and analogously for product effects.
Finally, the condition  (\ref{ii}) is satisfied for the alternative structure
$\langle \widetilde{K_{12}}, \widetilde{L_{12}}, \widetilde{\mu_{12}}\rangle$ in $4$ dimensions
since this is the standard realization of the statistical structure over the Hilbert space
${\mathbb C}^{2} \otimes {\mathbb C}^{2\ast}$.
Nevertheless, both statistical structures are isomorphic since they are both realized over
a Hilbert space of dimension $4$.

We summarize our considerations in the following
\begin{prop}\label{P1}
There exist two different, isomorphic statistical structures
$\langle K_{12}, L_{12}, \mu_{12}\rangle$ and
$\langle \widetilde{K_{12}}, \widetilde{L_{12}}, \widetilde{\mu_{12}}\rangle$
satisfying Axiom \ref{AX} such that
\begin{eqnarray}\label{prod12K}
 K({\mathbb C}^{2}) \otimes  K({\mathbb C}^{2}) & \subset & K_{12},
 \widetilde{K_{12}} \subset B({\mathbb C}^{2}) \otimes B({\mathbb C}^{2}) \;,  \\
\label{prod12L}
  L({\mathbb C}^{2}) \otimes  L({\mathbb C}^{2}) & \subset & L_{12}, \widetilde{L_{12}} \subset B({\mathbb C}^{2}) \otimes B({\mathbb C}^{2})
  \;, \mbox{ and }\\
  \mu_{12} \left(W_{12},F_{12}\right)&=&\mbox{Tr}\left(W_{12}F_{12} \right), \;
  \widetilde{\mu_{12}} \left(\widetilde{W_{12}},\widetilde{F_{12}}\right)=\mbox{Tr}\left(\widetilde{W_{12}}\widetilde{F_{12}} \right)
  \;.
\end{eqnarray}
\end{prop}

%%%%%%%%%%%%%%%%%%%%%%%%%%%%%%%%%%%%%%%%%%%%%%%%%%%%%%%%%%%%%%%%%%%%%%%%%%%%%%%%%%%%%%%%%%%%%%%%%%%%%%%%%%%%%%%%%%%%%%%%%%%%%%%
\section{Orientation of state spaces}\label{sec:OSS}
%%%%%%%%%%%%%%%%%%%%%%%%%%%%%%%%%%%%%%%%%%%%%%%%%%%%%%%%%%%%%%%%%%%%%%%%%%%%%%%%%%%%%%%%%%%%%%%%%%%%%%%%%%%%%%%%%%%%%%%%%%%%%%%

Let us explain the notion of a ``quantum orientation" once again in the context of this work,
simplified by restricting ourselves to finite-dimensional systems.
The product of two Hermitian operators $X,Y\in B({\mathcal H})$ is generally not Hermitian,
but the Jordan product $X\circ Y=1/2(XY+YX)$ is. It is plausible that the Jordan product can be
uniquely defined within the framework of statistical structure.
In fact, Alfsen and Shultz have constructed a Jordan algebra within the framework of a theory
that bears a strong resemblance to Ludwig's foundation of QT and is similarly motivated by physical questions.
We will briefly outline this construction: One can define ``projectors"
as extremal points of the convex set $L$ and consider decompositions of elements $X\in B$
(or $B'$ in the infinite-dimensional case)
as linear combinations of orthogonal projectors.
This provides a spectral theorem and hence kind of functional calculus for $B$, and allows one to define
$X \circ Y:=1/2( (X+Y)^2 - X^2 -Y^2)$.
In the standard realization of a statistical structure,
this results in the Jordan product of Hermitian operators.
The question of physical interpretation of the Jordan product will be left aside here.

In any case, the Jordan product is invariant under unitary or
anti-unitary automorphisms of the statistical structure.
In contrast, there is the associative product of operators of the
(in our case finite-dimensional) $C^\ast$-algebra
${\mathcal A}({\mathcal H}):=B({\mathcal H})\oplus {\sf i}B({\mathcal H})$.
Under an anti-unitary automorphism, the product $XY$ transforms into the product $YX$.
This highlights the need for a further structure in order to make the
extension of the Jordan algebra $B({\mathcal H})$ to the $C^\ast$-algebra
${\mathcal A}({\mathcal H})$ unique.

There are two equivalent definitions of this additional structure in the works of Alfsen and Shultz.
In this section, we examine the geometric concept of quantum orientation, drawing heavily on \cite{AHS80},
although we make some minor modifications that are appropriate for this work.
First, we recall the one-to-one correspondence between two-dimensional subspaces ${\mathcal H}_2$
and three-dimensional faces $F$ of $K({\mathcal H})$ such that $F$ consists of a statistical
operators $W\in K({\mathcal H})$ with support contained in ${\mathcal H}_2$.
Therefore, there exists a bijection
between the set ${\mathcal F}_3$ of all $3$-dimensional faces and  the Grassmann manifold $\mbox{Gr}(2,{\mathcal H})$
of all two-dimensional subspaces of ${\mathcal H}$.
We will add a few remarks on the structure of $\mbox{Gr}(2,{\mathcal H})$ although this is well-known, see, e.~g,, \cite{GH94}.

We fix an arbitrary two-dimensional subspace ${\mathcal H}_2^0$ and
represent every other two-dimensional subspace ${\mathcal H}_2$ as the image of
${\mathcal H}_2^0$ under some unitary operator $U\in {\mathcal U} ({}\mathcal H)$,
that is, ${\mathcal H}_2=U({\mathcal H}_2^0)$.
However, $U$ is not uniquely determined, but only up to right multiplication
by another unitary operator $V\in {\mathcal U} ({}\mathcal H)\cong {\mathcal U}(N)$ that leaves ${\mathcal H}_2^0$
- and thus also its orthogonal complement $\left({\mathcal H}_2^0\right)^\perp$ - invariant.
The subgroup of all such unitary operators $V$ is isomorphic to
${\mathcal U}\left( {\mathcal H}_2^0\right) \times {\mathcal U}\left( \left({\mathcal H}_2^0\right)^\perp\right)
\cong {\mathcal U}(2) \times {\mathcal U}(N-2)$.
These considerations lead to the well-known bijection
$\mbox{Gr}(2,{\mathcal H})\cong {\mathcal U}(N)/({\mathcal U}(2)\times {\mathcal U}(N-2))$
which can be used to endow $\mbox{Gr}(2,{\mathcal H})$ with the structure of a $4(N-2)$-dimensional compact
$C^\infty$-manifold. In what follows we will, however, only make use of the topology of $\mbox{Gr}(2,{\mathcal H})$.

Let $E^3$ denote the  unit ball in $3$-dimensional Euclidean space,
where the latter will be identified with ${\mathbb R}^3$.
For $K=K({\mathcal H})$ each face generated by two different extremal points of $K$ is affinely isomorphic to $E^3$.
One says that $K({\mathcal H})$ has the ``3-ball property".
Let ${\mathcal B}$  denote the set of all affine isomorphisms
$\phi$ from $E^3$ onto a face of $K$, that is,
\begin{equation}\label{defB}
 {\mathcal B} = \{ \phi: E^3 \to F\subset K({\mathcal H})\left|\right.
 F\in {\mathcal F}_3 \;\mbox{and}\, \phi \;\mbox{is an affine isomorphism}\}
 \;.
\end{equation}
${\mathcal B}$ will be equipped with the topology of pointwise convergence.
The group $O(3)$ operates freely and transitively on ${\mathcal B}$ by right multiplication.
Hence the continuous projection $\pi: {\mathcal B}\to {\mathcal F}_3,\;\pi(\phi):=F,$
can be viewed as a principal $O(3)$ bundle. For the local triviality of this bundle see Appendix \ref{LT}.

Next we note that $SO(3)$ is a normal subgroup of $O(3)$ such that $O(3)/SO(3)\cong {\mathbb Z}_2$.
By factorizing over  $SO(3)$ we thus obtain the principal  ${\mathbb Z}_2-$ bundle
\begin{equation}\label{defZ2}
 \Pi: {\mathcal B}/SO(3) \to {\mathcal F}_3
 \;.
\end{equation}
A global section of this bundle will be called a ``quantum orientation"
on $K({\mathcal H})$.
It can be shown that there exist exactly two global sections of (\ref{defZ2}) (see Appendix \ref{GS}), and
the additional structure of QT (additional to the statistical structure) consists in
specifying one of the two possible quantum orientations.

The above definitions can be generalized to a general convex set $K$ to a large extent; see \cite{AHS80},
and result in a ${\mathbb Z}_2$-bundle $\Psi: {\mathcal B}/SO(3) \rightarrow {\mathcal F}_3$.
The convex set $K$ is said to be {\it orientable} iff  $\Psi$ is globally trivial;
and in this case a continuous cross section of $\Psi$ is called an {\it orientation} of $K$.
It can be shown that if $K$ is the state space of a JB-algebra $A$, then $A$ is the
self-adjoint par of a $C^\ast$-algebra iff $A$ has the $3$-ball property and is orientable,
see theorem 8.4 of \cite{AHS80}, where the terms used above are defined in more detail.
In this case
there exists a $1:1$-correspondence between $C^\ast$-structures on $A \oplus {\sf i}A$
and orientations on $K$.

%%%%%%%%%%%%%%%%%%%%%%%%%%%%%%%%%%%%%%%%%%%%%%%%%%%%%%%%%%%%%%%%%%%%%%%%%%%%%%%%%%%%%%%%%%%%%%%%%%%%%%%%%%%%%%%%%%%%%%%%%%%%%%%
\section{Noether structures I}\label{sec:NSI}
%%%%%%%%%%%%%%%%%%%%%%%%%%%%%%%%%%%%%%%%%%%%%%%%%%%%%%%%%%%%%%%%%%%%%%%%%%%%%%%%%%%%%%%%%%%%%%%%%%%%%%%%%%%%%%%%%%%%%%%%%%%%%%%

In the case of statistical structures in the standard realization,
the automorphisms of these structures form a Lie group the identity component
of which can be identified with $U(N)/U(1)$.
To clarify this, first note that the automorphisms generated by
anti-unitary operators do not belong to the connected component of the identity
and hence do not contribute to its Lie algebra.
The reason for the quotient formation is that the
operation of unitary operators $U$ on statistical operators (or, analogously,
on effects) is given by $W\mapsto U W U^\ast$ and therefore constant phase factors
$\exp({\sf i}\alpha)\in U(1)$ act as identity. The Lie algebra ${\mathcal L}$
of the automorphism group can therefore be identified with the real
$N^2-1$-dimensional space of anti-Hermitian operators with trace $0$ and the
commutator as the Lie product. ${\mathcal L}$ is thus isomorphic to $su(N)$ and, as a vector space,
to the space $B_0({\mathcal H})$ of Hermitian operators with vanishing trace.

One possible physical interpretation of $B_0({\mathcal H})$ is given
by the set of (trace-free) observables.
(For the generalized concept of observables in Ludwig's approach, see below.)
Thus, in QT there is an isomorphism between trace-free observables
and the generators of one-parameter automorphism groups.
An example of this is the connection between the energy observable $H$ (Hamilton operator)
and the time evolution $U_t=\exp(-\frac{\sf i}{\hbar}\,H t)$ of a system,
which is formulated in the Schr\"odinger equation
Here, a specific orientation of the time axis (``arrow of time") is assumed,
and the time parameter $t$ used above reflects this orientation
in that it increases from the past toward the future.
Other examples are momentum observables and spatial translation (which, however,
requires an infinite-dimensional Hilbert space) as well as angular momentum observables
(or, in the finite-dimensional case, spin observables) and the representation of spatial rotations.
Based on Noether's theorem, which formulates a connection between symmetries and conserved quantities,
we will refer to the aforementioned isomorphism as a ``Noether structure."
In classical mechanics, the analogous Noether structure is given by $H \mapsto X_H$
with $\omega(X_H,.)=-dH$ and the symplectic $2$-form $\omega$, where $H$ is the Hamiltonian and
$X_H$ is called the ``Hamiltonian vector field" defined on the phase space, see, for example, \cite{A78}.

The idea of the Noether structure in QT can be found in the works of Alfsen and Shultz
\cite{AS98a,AS98b,AS12a,AS12b} under the name ``dynamical correspondence,"
since the authors primarily had in mind the above-mentioned example of the correspondence
between energy observables and time evolution.
We will briefly review the corresponding definitions, referring to the cited literature for details.
The framework for dynamical correspondences is the theory of unital JB-algebras $A$, Jordan algebras
with Jordan product $a \circ b$ and unit $1$  that are Banach spaces.
A bounded linear operator $\delta$ on $A$ is called an ``order derivation" if
$\tau \mapsto \exp(\tau \delta)$ is a one-parameter group of order isomorphisms.
Special cases are ``self-adjoint" order derivations of the form $\delta_b (a) = b \circ a$ for some
$b\in A$ and ``skew" order derivations satisfying $\delta(1)=0$.
A ``dynamical correspondence" is a
linear map $\psi: a \mapsto \psi_a$ from $A$ into the set of skew order
derivations on $A$, which satisfies the requirements:
\begin{eqnarray}
\label{requirement1}
  [\psi_a, \psi_b] &=&-[\delta_a, \delta_b]\quad \mbox{for all } a,b\in A\;,\\
  \label{requirement2}
  \psi_a(a) &=& 0\quad \mbox{for all } a\in A
  \;.
\end{eqnarray}
It can be shown that there exists a $1:1$ correspondence between dynamical correspondences $\psi$
and Jordan-compatible $C^\ast$-products on $A\oplus{\sf i}A$ such that $\psi_a(b)=\frac{\sf i}{2}(ab -ba)$,
see theorem 1 in \cite{AS98b}. Together with the cited theorem in section \ref{sec:OSS} this gives
also a $1:1$-correspondence between orientations on $K$ and dynamical correspondences on $A$.

%%%%%%%%%%%%%%%%%%%%%%%%%%%%%%%%%
\section{Generalized observables}\label{sec:GO}
%%%%%%%%%%%%%%%%%%%%%%%%%%%%%%%%%%%%%%%%%%%%%%%%%%%%%%%%%%%%%%%%%%%%%%%%%%%%%%%%%%%%%%%%%%%%%%%%%%%%%%%%%%%%%%%%%%%%%%%%%%%%%%%

We will outline how the concept of Noether structure can be incorporated into Ludwig's reconstruction of QT.
First, let us consider the concept of observables.

In the usual formalism of QT,
self-adjoint operators are physically interpreted as observables. In the finite-dimensional case
these are Hermitian operators $A$ with the spectral representation $A=\sum_{n} a_n P_n$.
Equivalent to specifying $A$ is therefore the definition of the possible measuring values $a_n$
together with the projectors $P_n$ that satisfy the constraint $\sum_{n} P_n = {\mathbbm 1}$.
The probability that the measuring value $a_n$ will be obtained is then given by
$\mbox{Tr }W P_n$ with the statistical operator $W$. Here, the projector $P_n$ plays the role of an ``effect."
It therefore makes sense to generalize the concept of an observable in such a way that
(in the finite-dimensional case) an observable is given by a finite sequence of measuring values $a_n, n=1,\ldots,M,$
and by a corresponding set of effects $F_n, n=1,\ldots,M,$ that satisfy the constraint $\sum_{n=1}^M F_n={\mathbbm 1}$.
In the special case where all $F_n$ are projectors, the observable will be called a ``decision observable".
In the general case we can define a Hermitian operator $A_1$ by $A_1=\sum_n a_n F_n$, but in general,
$A_1$ does not allow for the reconstruction of the observable $A$, except in the case of decision observables. $A_1$ is
rather considered as the ``first moment" of the observable $A$. This yields another interpretation
of the space $B({\mathcal H})$ of Hermitean operators as the space of first moments of (generalized) observables,
which is suited for the consideration of Noether structures.\\

For the purposes of this paper, it seems appropriate to extend the concept of generalized observables
so that it can be applied to the case of dimensional quantities.
In doing so, we draw on the proposal in \cite{S25} for the formal treatment of physical dimensions.
There the ``range of a physical quantity" ${\mathcal R}$ was defined as a one-dimensional, real, ordered
vector space and ${\mathcal R}_{\ge 0} \subset {\mathcal R}$ the subset of non-negative values.
A physical unit is a positive basis vector $r_0\in {\mathcal R}_{\ge 0}$,
such that any value $r\in {\mathcal R}$ of the corresponding quantity  can be represented by a real
number via $r=\xi\,r_0,\; \xi\in{\mathbb R}$.

Thus we assume a generalized observable $A$ with values in  ${\mathcal R}$  to be given
by a sequence of effects $F_n,\; n=1,\ldots , M$ and a corresponding sequence
of measuring values $a_n\in {\mathcal R},\; n=1,\ldots,M$. The first moment of $A$ will be given by
\begin{equation}\label{1stmom}
 A_1= \sum_{n=1}^{M} a_n\, F_n =\sum_{n=1}^{M} \xi_n\,r_0\, F_n 
 =\left(\sum_{n=1}^{M} \xi_n\, F_n \right)\otimes r_0 \in B({\mathcal H})\otimes{\mathcal R}
 \;,
\end{equation}
where $r_0\in {\mathcal R}_{\ge 0}$ is some physical unit and $a_n = \xi_n\,r_0$ the corresponding decomposition
of the $n$-th measuring value.

It will be in order to reformulate the assumptions on the future-directed time evolution of quantum states
taking into account physical dimensions.
This will be of the form
\begin{equation}\label{timeevolution}
 W\in K({\mathcal H}) \mapsto U_t\,W \, U_t^\ast=: \exp(\underbrace{- {\sf i} \textsf{H}\,t}_{\in {\mathcal L}})\,W
 \;,
\end{equation}
where we have assumed that $t>0$ corresponds to the future time evolution and, as in Section \ref{sec:NSI},
${\mathcal L}\cong {\sf i}B_0({\mathcal H})$ denotes the Lie algebra of the automorphism group of the quantum system which
we consider to be dimensionless. If the value $t$ in (\ref{timeevolution}) is not just a real parameter
but a physical quantity of dimension ``time", say, $t\in{\mathcal T}$, then the remaining factor $- {\sf i} \textsf{H}$
in  (\ref{timeevolution}) must be an element of ${\mathcal L}\otimes {\mathcal T}^{-1}$, or, using the above
identification, $- {\sf i} \textsf{H}\in  {\sf i}B_0({\mathcal H}) \otimes {\mathcal T}^{-1}$.
The range ${\mathcal T}^{-1}$ of the quantity ``inverse time" can be identified with the dual space
 ${\mathcal T}^\ast$, see \cite{S25}.

%%%%%%%%%%%%%%%%%%%%%%%%%%%%%%%%%%%%%%%%%%%%%%%%%%%%%%%%%%%%%%%%%%%%%%%%%%%%%%%%%%%%%%%%%%%%%%%%%%%%%%%%%%%%%%%%%%%%%%%%%%%%%%%
\section{Noether structures II}\label{sec:NSII}
%%%%%%%%%%%%%%%%%%%%%%%%%%%%%%%%%%%%%%%%%%%%%%%%%%%%%%%%%%%%%%%%%%%%%%%%%%%%%%%%%%%%%%%%%%%%%%%%%%%%%%%%%%%%%%%%%%%%%%%%%%%%%%%
The Noether structure is an axiomatization of the double role of, say, Hamiltonians as energy observable
and generators of time evolution, and analogously for the other examples mentioned above.
If we stick with the example of time evolution  and keep in mind the considerations
regarding physical dimensions from the previous section, the Noether structure would be represented by a linear map
\begin{equation}\label{Psi}
 \Psi: B({\mathcal H}) \otimes {\mathcal E} \rightarrow {\sf i}\, B_0({\mathcal H}) \otimes {\mathcal T}^{-1}
 \;,
\end{equation}
where ${\mathcal E}$ denotes the range of energies.
It can be factored according to
\begin{eqnarray}
\label{Psi1}
 \Psi &=& \psi \otimes \hbar^{-1},\quad \mbox{where} \\
 \label{Psi2}
 \psi &:&  B({\mathcal H}) \rightarrow {\sf i}\, B_0({\mathcal H})) \quad \mbox{and}\\
 \label{Psi3}
 \hbar^{-1}&:& {\mathcal E}   \rightarrow {\mathcal T}^{-1}
 \;.
\end{eqnarray}
Here we have already denoted the conversion factor by $\hbar^{-1}$  thereby anticipating the desired result.
However, it is clear that the factorization $\Psi= \psi \otimes \hbar^{-1}$ is only unique
up to a factor in $\psi$ and the inverse factor in $\hbar^{-1}$. Therefore, we can only
fix the value of $\hbar^{-1}$ if we simultaneously fix the absolute value of $\psi$, as it has been done
in the approach of Alfsen and Shultz, see (\ref{requirement1}).
To this end we do not want to invoke the theory of Jordan algebras and rather use the fact
that the Lie algebra ${\mathcal L}\cong su(N)$
has a natural Euclidean structure essentially given by its Killing form
$\beta(X,Y)=- 2 N \mbox{Tr}(X\,Y)$, see, e.~g., \cite {H72}.

In what follows we only consider the ``dimensionless" variant of the Noeter structure and  define
\begin{defi}\label{D1}
A {\it Noether structure} is a linear map $\psi: B({\mathcal H})\to {\mathcal L}$ subject to the following requirements
\begin{eqnarray}
\label{cond1}
  \psi(X) &=& 0 \Leftrightarrow X= \alpha {\mathbbm 1}  \quad \mbox{for all }X \in B \mbox{ and some } \alpha\in {\mathbb R} \\
  \label{cond2}
\exp(t \psi(X))(X) &=& X \quad \mbox{for all } X\in B \mbox{ and } t\in {\mathbb R}\\
 \label{cond3} \beta(\psi(X),\psi(Y)) &=&- 2 N \mbox{Tr}(X\,Y)\quad \mbox{for all}\; X, Y\in  B({\mathcal H})
\;.
\end{eqnarray}
\end{defi}
The first condition (\ref{cond1}) expresses the fact that the phase factors operate as the identity
on ensembles and effects. The second condition corresponds to (\ref{requirement2}) and says that
the one-parameter group generated by $\psi(X)$ leaves $X$ fixed. Finally, the third condition replaces
(\ref{requirement1}), cf.~ the discussion above.
Then we can the prove, see \ref{AP2}, the following
%%%%%%%%%%%%%%%%%%%%%%%%%%%%%%%%%%%%%%%%%%%%%%%%%%%%%%%%%%%%%%%%%%%%%%%%%%%%%%%%%%%%%%%%%%%%%%%%%%%
\begin{prop}\label{P2}
Consider a statistical structure in the standard representation
$\langle  K({\mathcal H}), L({\mathcal H}), \mu({\mathcal H})\rangle$ and a Noether structure
$\psi: B({\mathcal H}) \to {\mathcal L}$, where ${\mathcal L}$
is identified with the space ${\sf i}B_0({\mathcal H})$ of anti-hermitean operators with vanishing trace acting on ${\mathcal H}$.
Then $\psi$ is of the form $\psi(X) = \pm{\sf i} X $ for all $X\in  B({\mathcal H})$.
\end{prop}
%%%%%%%%%%%%%%%%%%%%%%%%%%%%%%%%%%%%%%%%%%%%%%%%%%%%%%%%%%%%%%%%%%%%%%%%%%%%%%%%%%%%%%%%%%%%%%%%%%%

This proposition shows that there exist exactly two different Noether structures, analogously to
two different quantum orientations. Empirical evidence suggests the choice
$\psi(X) = -{\sf i}X$ or, more explicitly,
\begin{equation}\label{Noetherchoice}
 \psi(X) (Y) = - {\sf i} [X,Y]\quad \mbox{for all}\; X,Y \in  B({\mathcal H})
 \;.
\end{equation}

Finally, we note that extending the statistical structure $\langle K,L,\mu\rangle$ to a
``statistical Noether structure" $\langle K,L,\mu,\psi\rangle$ reduces the
automorphism group of this structure to $SU(N)$, assuming that Axiom \ref{AX} holds.
Anti-unitary transformations transform $\psi$ into $-\psi$
and are therefore to be excluded as automorphisms.
Similarly the choice of Hilbert space
${\mathcal H}$ is restricted only up to unitary isomorphisms.
This means that the tensor product of Hilbert spaces is also unique
up to unitary isomorphisms, and the composition of systems can be handled
without the problems described in section $\ref{sec:CS}$.

Or, to put it another way, using the Noether structure $\psi$, we can,
as in the works of Alfsen and Shultz, define a $C^\ast$-product in
$B({\mathcal H})\oplus {\sf i}B({\mathcal H})$ 
and thus uniquely define ``positive" elements
of $B({\mathcal H}_1)\otimes B({\mathcal H}_2)$ as squares
w.~r.~t.~this product. This yields the unique state space of the composite system
$K_{12}\subset B({\mathcal H}_1) \otimes B({\mathcal H}_1)$
we have been looking for in section \ref{sec:CS}.

%%%%%%%%%%%%%%%%%%%%%%%%%%%%%%%%%%%%%%%%%%%%%%%%%%%%%%%%%%%%%%%%%%%%%%%%%%%%%%%%%%%%%%%%%%%%%%%%%%%%%%%%%%%%%%%%%%%%%%%%%%%%%%%
\section{Completely positive transformations}\label{sec:CPT}
%%%%%%%%%%%%%%%%%%%%%%%%%%%%%%%%%%%%%%%%%%%%%%%%%%%%%%%%%%%%%%%%%%%%%%%%%%%%%%%%%%%%%%%%%%%%%%%%%%%%%%%%%%%%%%%%%%%%%%%%%%%%%%%

%%%%%%%%%%%%%%%%%%%%%%%%%%%%%%%%%%%%%%%%%%%%%%%%%%%%%%%%%%%%%%%%%%%%%%%%%%%%%%%%%%%%%%%%%%%%%%%%%%%%%%%%%%%%%%%%%%%%%%%%%%%%%%%
\subsection{Generalities}\label{sec:CPT1}
%%%%%%%%%%%%%%%%%%%%%%%%%%%%%%%%%%%%%%%%%%%%%%%%%%%%%%%%%%%%%%%%%%%%%%%%%%%%%%%%%%%%%%%%%%%%%%%%%%%%%%%%%%%%%%%%%%%%%%%%%%%%%%%

Further evidence that the statistical structure alone is insufficient to capture
essential features of the QT can be found in the treatment of general
state transformations, which go beyond the time evolution of closed systems
and also include state transformations caused by measurements.
Affine transformations $T: K \to K$ would be suitable for this purpose.
However, it has been found that this term is too general
and that all physical examples involve so-called ``completely positive transformations."
Their definition involves tensor products or multiplication of operators, see, e.~g., \cite{K83},
and hence cannot be given on the basis of purely statistical structures.

Recall the following definitions adapted to the finite-dimensional case:
\begin{defi} \label{defcp}
\begin{enumerate}
  \item Let ${\mathcal H}$ be an $N$-dimensional Hilbert space. Then a linear map
    ${T}: B({\mathcal H})  \to B({\mathcal H})$ is called ``monotone" iff
    it maps $B^+({\mathcal H})$ into $B^+({\mathcal H})$.
  \item A linear map $T:B({\mathcal H}) \to B({\mathcal H})$ is called ``completely positive"
    iff it is monotone and for all $n\in {\mathbb N}$  the map
    $T \otimes {\mathbbm 1}_{B\left({\mathbb C}^n\right)}: B({\mathcal H}) \otimes B({\mathbb C}^n) \to  B({\mathcal H}) \otimes  B({\mathbb C}^n)$
   is also monotone.
\end{enumerate}
\end{defi}

It can be shown that an equivalent characterization of ``completely positive maps" $T$ is given by
the possibility to represent $T$ in the form
\begin{equation}\label{Kraus}
  T(X) = \sum_{i=1}^{n} A_i\,X \,A_i^\ast  \quad \mbox{for all } X\in B({\mathcal H})
  \;,
\end{equation}
where the ``Kraus operators" $A_i$ are complex linear operators $A_i: {\mathcal H}\to  {\mathcal H}$
and $1\le n \le N^2$, see, e.~g., \cite{NC00}.
Let the space ${\mathcal A}({\mathcal H})$ of all linear operators $X:{\mathcal H}\to {\mathcal H}$ be identified with
$B({\mathcal H})\oplus {\sf i}B({\mathcal H})$, then it is obvious that
a completely positive map $T$ can be defined on ${\mathcal A}({\mathcal H})$ by means of complex-linear extension,
and that this extension has the same Kraus operator sum representation as (\ref{Kraus}).
Linear, monotonic mappings ${T}: B({\mathcal H})  \to B({\mathcal H})$ are also referred to as ``operations."
This terminology dates back to works such as \cite{D69}, whereas in more recent literature, the
term ``operation" is often used exclusively for completely positive mappings.
It is convenient not to include the property of trace-preserving into this definition,
which implies that the operation $T$ in the case of (\ref{Kraus}) is written as a sum of $n$ other operations.

Further consider the following definition:
\begin{defi}\label{D5}
An operation  $T:B({\mathcal H})\to B({\mathcal H})$ is called ``pure" iff
$T$ maps positive rank $1$ operators onto  non-negative rank $1$ operators.
Equivalently, we may define ``pure" operations $T$ by the condition that for every one-dimensional
projector $P_\varphi$, with $\varphi\in{\mathcal H}$ and $\|\varphi\|=1$, $T \left(P_\varphi\right)$ will be of the form
$T \left(P_\varphi\right)= \lambda P_\psi$ with $\psi\in{\mathcal H},\;\|\psi\|=1$ and $\lambda\ge 0$.
\end{defi}

Pure operations can be characterized by the following theorem, which also holds in the case of
infinite-dimensional Hilbert spaces, see \cite{D69}, theorem 3.1:
\begin{theorem}\label{TPO}
  Let $T:B({\mathcal H})\to B({\mathcal H})$ be a pure operation. Then
  \begin{enumerate}
    \item[(i)] there exists a complex linear operator $A: {\mathcal H} \to {\mathcal H}$ such that
    $T(X)=A X A^\ast$ for all $X\in B({\mathcal H})$, or
    \item[(ii)] there exists a complex anti-linear operator $B: {\mathcal H} \to {\mathcal H}$ such that
    $T(X)=B X B^\ast$ for all $X\in B({\mathcal H})$, or
    \item[(iii)] there exists a $W\in B^+({\mathcal H})$ and a $\psi\in{\mathcal H},\;\|\psi\|=1,$ such that
    $T(X)= \mbox{Tr}(W X)\,P_\psi$ for all $X\in B({\mathcal H})$.
  \end{enumerate}
\end{theorem}
Since we will frequently refer to cases (i) through (iii), we define:
\begin{defi}\label{Dcases}
 A pure operation $T:B({\mathcal H})\to B({\mathcal H})$ is said to be ``of type (i), (ii), or (iii)"
iff it has the form (i), (ii), or (iii), respectively, according to the list of cases of Theorem \ref{TPO}.
\end{defi}

The list of cases (i), (ii), and (iii) is exhaustive but not exclusive,
as suggested by the wording ``either" in \cite{D69}, theorem 3.1.
In the case (i) of this theorem the pure operation is completely positive, since $T(X)=A X A^\ast$
is a special Kraus sum operator representation of $T$. In the case (iii) we may use the spectral
representation $W=\sum_{i=1}^N w_i |i\rangle \langle i|$ of $W\ge 0$ satisfying $w_i\ge 0$ for all $i=1,\ldots,N$
and conclude
\begin{eqnarray}
\label{case3a}
  T(X) &=& \mbox{Tr }(W X)\,P_\phi \\
  \label{case3b}
   &=& \sum_i w_i\, \mbox{Tr} ( |i\rangle \langle i| X)\,  |\phi\rangle \langle \phi|\\
   \label{case3c}
   &=& \sum_i \underbrace{\sqrt{w_i}\, |\phi\rangle \langle i|}_{=:A_i}
   X \underbrace{|i\rangle \langle \phi|\, \sqrt{w_i}}_{=A_i^\ast}\\
   \label{case3d}
   &=& \sum_i A_i\,X\,A_i^\ast
   \;,
\end{eqnarray}
which shows that also in this case $T$ will be completely positive.
The example $T(X) = P_\phi \,X\,P_\phi= \mbox{Tr}(P_\phi X)\,P_\phi$
belongs to both cases, (i) and (iii).
In the case (ii) the pure operation can be both, completely positive or not.
Let $c:{\mathcal H}\to {\mathcal H}$ denote complex conjugation (w.~r.~t.~some orthonormal basis).
Then $T(X) =  c\,P_\phi \,X\,P_\phi\,c^\ast=
c\,|\phi\rangle \langle \phi| X|\phi\rangle \langle \phi| c^\ast=
\mbox{Tr}(P_\phi\,X)\,P_{\overline{\phi}}$
represents an example belonging to all three cases (i), (ii) and (iii).
We summarize our considerations in the following
%%%%%%%%%%%%%%%%%%%%%%%%%%%%%%%%%%%%%%%%%%%%%%%%%%%%%%%%%%%%%%%%%%%%%%%%%%%%%%%%%%%%%%%%%%%%%%%%%%%%%%%%
\begin{prop}\label{PC}
 Every pure operation  of type (i) or (iii) is completely positive.
\end{prop}

%%%%%%%%%%%%%%%%%%%%%%%%%%%%%%%%%%%%%%%%%%%%%%%%%%%%%%%%%%%%%%%%%%%%%%%%%%%%%%%%%%%%%%%%%%%%%%%%%%%%%%%%

Next we assume a standard realization of a statistical Noether structure
$\langle K({\mathcal H}),L({\mathcal H}),\mu({\mathcal H}),\psi({\mathcal H})\rangle$ and
consider an arbitrary $3$-dimensional face $F\in{\mathcal F}_3$ corresponding to
a two-dimensional subspace ${\mathcal H}_2$. Let $(\varphi,\psi)$ be any orthonormal
basis in ${\mathcal H}_2$, and $B(F)\cong B({\mathcal H}_2)$ be the $4$-dimensional
subspace of $B({\mathcal H})$ spanned by all statistical operators $W\in F$ or,
equivalently, consisting of all Hermitean operators with support included in ${\mathcal H}_2$.
We define the following orthonormal basis in $B(F)$:
\begin{eqnarray}\label{sigma0}
 \sigma_0(F)&=&\frac{1}{\sqrt{2}}\left(|\varphi\rangle \langle \varphi | +|\psi\rangle \langle \psi | \right),\\
 \label{sigma1}
 \sigma_1(F)&=&\frac{1}{\sqrt{2}}\left(|\varphi\rangle \langle \psi | +|\psi\rangle \langle \varphi| \right),\\
 \label{sigma2}
 \sigma_2(F)&=&\frac{{\sf i}}{\sqrt{2}}\left(-|\varphi\rangle \langle \psi | +|\psi\rangle \langle \varphi| \right),\\
 \label{sigma3}
 \sigma_3(F)&=&\frac{1}{\sqrt{2}}\left(|\varphi\rangle \langle \varphi | -|\psi\rangle \langle \psi | \right),
 \end{eqnarray}
analogously to $\sigma_i^\pm(F), \; i=0,1,2,3,$ in (\ref{sigmapm0}-\ref{sigmapm3}).
It will be referred to as the ``standard basis" in $B(F)$.
This gives rise to an
affine isomorphism $\phi:E^3 \to F$ in the following way:
\begin{eqnarray}
\label{defphi1}
  \tau_0(F)=\frac{1}{\sqrt{2}} \, \sigma_0(F),\; \tau_i(F)&=& \tau_0(F) +\frac{1}{\sqrt{2}} \,\sigma_i(F),\quad \mbox{for}\; i=1,2,3, \\
  \label{defphi2}
  \phi(0,0,0) &=& \tau_0(F),\\
   \label{defphi3}
  \phi(1,0,0) &=& \tau_1(F),\\
   \label{defphi4}
  \phi(0,2,0) &=& \tau_2(F),\\
   \label{defphi5}
  \phi(0,0,3) &=& \tau_3(F)
  \;,
\end{eqnarray}
compare (\ref{defphi1} - \ref{defphi5}).
Here we have used the identification of $E^3$ with the unit ball in ${\mathbb R}^3$ and that an affine
isomorphism $\phi:E^3 \to F$ is uniquely determined by the image of four point that affinely generate ${\mathbb R}^3$.
The global section $\Sigma: {\mathcal F}_3\to {\mathcal B}/SO(3) $
of the ${\mathbb Z}_2$-bundle $\Pi:{\mathcal B}(SO(3)\to {\mathcal F}_3$ satisfying
$\Sigma(F) = \left\{\phi\,R\left|\right. R\in SO(3)\right\}$ defines the corresponding quantum
orientation which will be held fixed in what follows.

Next we consider, analogously to the well-known characterization of the determinant:
\begin{defi}\label{Ddet}
  For all $F\in {\mathcal F}_3$ we define a map $\mbox{Det}:B(F)^4 \to {\mathbbm R}$ by the three postulates
  \begin{enumerate}
    \item $\mbox{Det}$ is 4-linear,
    \item $\mbox{Det}$ is alternating, and
    \item $\mbox{Det}(\sigma_0(F), \sigma_1(F),\sigma_2(F),\sigma_3(F))=1$.
  \end{enumerate}
\end{defi}
Recall, that, by definition, an alternating map vanishes if two of its arguments are equal.
By 4-linearity it follows that $\mbox{Det}$ will change its sign if two of its arguments are swapped.
Further, it follows that $\mbox{Det}\left(X_0,X_1,X_2,X_3\right)=0$ if $\left(X_0,X_1,X_2,X_3\right)$ is linearly dependent
in $B(F)$.

\begin{defi}\label{Dop}
  An operation $T: B({\mathcal H}) \to B({\mathcal H})$ is called ``orientation-preserving"
  iff the following holds:
  \begin{enumerate}
    \item For all $F\in {\mathcal F}_3$ there exists an  $F'\in {\mathcal F}_3$ such that
    $T\left( B(F)\right) \subset B(F')$, and
    \item for all $F\in {\mathcal F}_3$ and $X_i\in B(F),\,i=0,1,2,3,$ the inequality
    $\mbox{Det}(X_0,X_1,X_2,X_3)\ge 0$ always implies $\mbox{Det}(T\,X_0,T\,X_1,T\,X_2,T\,X_3)\ge 0$.
  \end{enumerate}
\end{defi}

%%%%%%%%%%%%%%%%%%%%%%%%%%%%%%%%%%%%%%%%%%%%%%%%%%%%%%%%%%%%%%%%%%%%%%%%%%%%%%%%%%%%%%%%%%%%%%%%%%%%%%%%%%%%%%%%%
\begin{prop}\label{Pdet}
 Let $F\in {\mathcal F}_3$. We expand the operators $X_i\in B(F),\,i=0,1,2,3,$ w.~r.~t.~the
 orthonormal basis $(\sigma_0(F), \sigma_1(F),\sigma_2(F),\sigma_3(F))$ according to
 $X_i= \sum_{j=0}^{3}\xi_{ji}\sigma_j(F)$ and consider the $4\times 4-$matrix ${\boldsymbol \xi}$
 formed of the  expansion coefficients $\xi_{ij}$.
 Then the following holds:
 \begin{equation}\label{eqPdet}
 \mbox{Det}\left(X_0,X_1,X_2,X_3\right)=\det {\boldsymbol \xi}
 \;.
 \end{equation}
\end{prop}
%%%%%%%%%%%%%%%%%%%%%%%%%%%%%%%%%%%%%%%%%%%%%%%%%%%%%%%%%%%%%%%%%%%%%%%%%%%%%%%%%%%%%%%%%%%%%%%%%%%%%%%%%%%%%%%%%

{\bf Proof}: Inserting the expansion of the $X_i\in B(F),\,i=0,1,2,3,$ into the arguments of $\mbox{Det}$ yields
\begin{eqnarray}
\label{proofPdet1}
 \mbox{Det}\left(X_0,X_1,X_2,X_3\right)&=&  \mbox{Det}\left(\sum_{i=0}^{3}\xi_{i0}\sigma_i,
 \sum_{j=0}^{3}\xi_{j1}\sigma_j,\sum_{k=0}^{3}\xi_{k2}\sigma_k,\sum_{\ell=0}^{3}\xi_{\ell 3}\sigma_\ell\right) \\
 \label{proofPdet2}
   &=& \sum_{ijk\ell} \xi_{i0}\,\xi_{j1}\,\xi_{k2}\,\xi_{\ell 3}\,\mbox{Det}\left(\sigma_i,\sigma_j,\sigma_k,\sigma_\ell \right)\\
  \label{proofPdet3}
   &=& \sum_{ijk\ell} \xi_{i0}\,\xi_{j1}\,\xi_{k2}\,\xi_{\ell 3}\,\epsilon_{ijk\ell}^{0123}\,
   \underbrace{\mbox{Det}\left(\sigma_0,\sigma_1,\sigma_2,\sigma_3 \right)}_{=1}\\
   \label{proofPdet4}
   &=& \det {\boldsymbol \xi}
   \;,
\end{eqnarray}
where $\epsilon_{ijk\ell}^{0123}$ denotes the usual $\epsilon$-tensor with the values $\pm1$ if
$(ijk\ell)$ is an even/odd permutation of $(0123)$ and $0$ else.
\hfill$\square$\\

Thus Proposition \ref{Pdet} shows that, by using a coordinate representation,
$Det$ can be calculated as the usual determinant. This implies the following
%%%%%%%%%%%%%%%%%%%%%%%%%%%%%%%%%%%%%%%%%%%%%%%%%%%%%%%%%%%%%%%%%%%%%%%%%%%%%%%%%%%%%%%%%%%%%%%%%%%%%%%%%%%%%%%%%
\begin{cor}\label{C1}
Let $F\in {\mathcal F}_3$. For all operators $X_i\in B(F),\,i=0,1,2,3,$  it holds that
  $\mbox{Det}\left(X_0,X_1,X_2,X_3\right)=0$ iff $\left(X_0,X_1,X_2,X_3\right)$  is linearly
  dependent in $B(F)$.
\end{cor}
%%%%%%%%%%%%%%%%%%%%%%%%%%%%%%%%%%%%%%%%%%%%%%%%%%%%%%%%%%%%%%%%%%%%%%%%%%%%%%%%%%%%%%%%%%%%%%%%%%%%%%%%%%%%%%%%%

%%%%%%%%%%%%%%%%%%%%%%%%%%%%%%%%%%%%%%%%%%%%%%%%%%%%%%%%%%%%%%%%%%%%%%%%%%%%%%%%%%%%%%%%%%%%%%%%%%%%%%%%%%%%%%%%%
\begin{prop}\label{Tdet}
 Let $T: B({\mathcal H}) \to B({\mathcal H})$ be an operation such that for all
 $F\in {\mathcal F}_3$ there exists an $F'\in {\mathcal F}_3$ satisfying $T\left( B(F)\right)\subset B(F')$.
 Further, for all $F\in {\mathcal F}_3$, let ${\sf T}(F)$ be the $4\times 4-$matrix formed of the expansion
 coefficients ${\sf T}_{kj}$ according to
 \begin{equation}\label{Tdet1}
   T\,\sigma_j(F)=\sum_{k=0}^{3}{\sf T}_{kj}\, \sigma_k(F')
   \;.
 \end{equation}
 Then $T$ will be orientation-preserving iff $\det {\sf T}(F)\ge 0$ for all $F\in {\mathcal F}_3$.
\end{prop}
%%%%%%%%%%%%%%%%%%%%%%%%%%%%%%%%%%%%%%%%%%%%%%%%%%%%%%%%%%%%%%%%%%%%%%%%%%%%%%%%%%%%%%%%%%%%%%%%%%%%%%%%%%%%%%%%%

{\bf Proof}: Let $X_i\in B(F)$ for $i=0,1,2,3,$  and consider the expansions
\begin{eqnarray}
\label{Tdet2}
  T\,X_i &=& T\left(\sum_j \xi_{ji}\,\sigma_j(F) \right)= \sum_j  \xi_{ji}\,T\,\sigma_j(F)\\
  \label{Tdet3}
   &=&\sum_{jk} \xi_{ji}\, {\sf T}_{kj}\,\sigma_k(F') =
   \sum_k\left(\sum_j  {\sf T}_{kj}\,\xi_{ji}\right) \,\sigma_k(F')\\
 \label{Tdet4}
   &=& \sum_k \left( {\sf T}\,{\boldsymbol\xi}\right)_{ki}\,\sigma_k(F')
   \;,
\end{eqnarray}
where ${\sf T}\,{\boldsymbol\xi}$ denotes the matrix product.
Using  Proposition \ref{Pdet} and the determinant product theorem we conclude
\begin{equation}\label{Tdet5}
 \mbox{Det}\left(T\,X_0, T\,X_1,T\,X_2,T\,X_3\right)= \det \left( {\sf T}\,{\boldsymbol\xi}\right)
 = \det \left( {\sf T}\right)\,\det\left({\boldsymbol\xi}\right)
 =\det \left( {\sf T}\right)\, \mbox{Det} \left( X_0,X_1,X_2,X_3\right)
 \;.
\end{equation}
Hence $T$ is orientation-preserving iff $\det \left( {\sf T}\right)\ge 0$ .                 \hfill$\square$

%%%%%%%%%%%%%%%%%%%%%%%%%%%%%%%%%%%%%%%%%%%%%%%%%%%%%%%%%%%%%%%%%%%%%%%%%%%%%%%%%%%%%%%%%%%%%%%%%%%%%%%
\begin{prop}\label{PU1}
 Let $U:{\mathcal H}\to {\mathcal H}$ be a unitary operator, then the map $T_U: B({\mathcal H}) \to B({\mathcal H})$
 defined by $T_U(X)=U\,X\,U^\ast$ will be a pure, completely positive, orientation-preserving operation.
 \end{prop}
%%%%%%%%%%%%%%%%%%%%%%%%%%%%%%%%%%%%%%%%%%%%%%%%%%%%%%%%%%%%%%%%%%%%%%%%%%%%%%%%%%%%%%%%%%%%%%%%%%%%%%%

{\bf Proof}: Obviously, $T_U$ is a completely positive operation satisfying the first item of the definition \ref{Dop}
since it maps any face $F\in {\mathcal F}_3$ onto the face $F'=\{U\,X\,U^\ast \left|\right. X\in F\}\in {\mathcal F}_3$.
Moreover, $T_U\left(  \sigma_0(F)\right)= \sigma_0(F')$ and $T_U$ maps the three-dimensional Euclidean space $B_0(F)$
isometrically onto $B_0(F')$. Hence the matrix ${\sf T}$ corresponding to $T_U$ defined in Proposition \ref{Tdet}
has a block form with ${\sf T}_{00}=1$ and only the submatrix $\widetilde{\sf T}$
with entries ${\sf T}_{ij},\,1\le i,j \le 3,$ has possibly non-zero values. It follows that $\widetilde{\sf T}\in O(3)$
and hence $\det \widetilde{\sf T} =\pm 1$. The value $\det  \widetilde{\sf T} =-1$ can be excluded since the unitary operator
$U$ can be continuously deformed into the identity operator and $\det {\sf T}$ depends continuously on $U$.
It follows that $\det {\sf T}=1$ and that $T_U$ is orientation-preserving due to Proposition \ref{Tdet}.  \hfill$\square$

%%%%%%%%%%%%%%%%%%%%%%%%%%%%%%%%%%%%%%%%%%%%%%%%%%%%%%%%%%%%%%%%%%%%%%%%%%%%%%%%%%%%%%%%%%%%%%%%%%%%%%%%%%%%%%%%%%%%%%%%%%%%%%%
\subsection{Characterization of completely positive pure operations}\label{sec:CPT2}
%%%%%%%%%%%%%%%%%%%%%%%%%%%%%%%%%%%%%%%%%%%%%%%%%%%%%%%%%%%%%%%%%%%%%%%%%%%%%%%%%%%%%%%%%%%%%%%%%%%%%%%%%%%%%%%%%%%%%%%%%%%%%%%

We have already pointed out that the definition of completely positive
operations requires either the tensor product or the
associative product of operators and thus presupposes a quantum orientation
or an equivalent structure. It therefore stands to reason to seek
a characterization of completely positive operations that refers
directly to the orientation. In doing so, we have found the following partial result.

\begin{theorem}\label{Tpure}
Let $T:B({\mathcal H})  \to  B({\mathcal H})$ be a pure operation. Then $T$ is completely positive
iff it is orientation-preserving.
\end{theorem}
This theorem can be viewed as a generalization of theorem 4.35 in \cite{AS12a},
which states the analogous equivalence for unital order automorphisms.\\

{\bf Proof}: We consider the three cases of a pure operation described in Theorem \ref{TPO} separately.

\subsubsection{Case (i) } \label{Ci}

In this case $T$ is of the form $T \,X=A\,X\,A^\ast$ with
a complex linear operator $A: {\mathcal H} \to {\mathcal H}$ and is hence completely positive.
We have to show that $T$ is orientation-preserving.\\

We consider a general three-dimensional face $F\in{\mathcal F}_3$ corresponding to a two-dimensional
subspace ${\mathcal H}_2 \subset {\mathcal H} $.
The subspace $A\left( {\mathcal H}_2 \right)$ will be of dimension $0,1$ or $2$. In either case
$T(F)$ will be a subset of $B(F')$ for another three-dimensional face $F'\in{\mathcal F}_3$ and the
first item of definition \ref{Dop} will be satisfied.

Assume that the dimension of  $A\left( {\mathcal H}_2 \right)$ will be $0$ or $1$ and consider $X_i\in B(F),\,i=0,1,2,3,$
such that $\mbox{Det}(X_0,X_1,X_2,X_3)\ge 0$. Then $\dim T(B(F))<4$ and hence
$(T\,X_0,T\,X_1,T\,X_2,T\,X_3)$ will be linearly dependent
and hence $\mbox{Det}(T\,X_0,T\,X_1,T\,X_2,T\,X_3)=0$, see Corollary \ref{C1}.
In this case the condition $\mbox{Det}(T\,X_0,T\,X_1,T\,X_2,T\,X_3)\ge 0$
for $T$ being orientation-preserving will be satisfied.
Thus we may assume  $\dim \,A\left( {\mathcal H}_2 \right)=2$ for the following part of the proof.

Next we consider the polar decomposition $A=U\,P$ where  $U:{\mathcal H}\to {\mathcal H}$ is a
unitary operator and $P:{\mathcal H}\to {\mathcal H}$ a positive semi-definite one.
Let $T=T_2 \circ T_1$ be the corresponding decomposition of $T$ where $T_1(X)= P\,X\,P$
and $T_2(Y) = U\,Y\,U^\ast$. Since the property ``orientation-preserving" is conserved under
composition of maps, it suffices to show that $T_1$ and $T_2$ are orientation-preserving.
For $T_2$ this holds due to Proposition \ref{PU1}.
We can therefore restrict ourselves to the case $T_1(X)=P\,X\,P$.

We consider an eigenbasis of $P$ in the Hilbert space ${\mathcal H}$ and represent all
operators by matrices w.~r.~t.~this eigenbasis without changing the notation.
Similarly, we set ${\mathcal H}={\mathbb C}^N$.
Thus we may write $P=\mbox{diag}(p_1,p_2,\ldots,p_N)$ with $p_n\ge 0$ for $n=1,\ldots, N$.
$P$ can be written as a product of diagonal matrices with only a single diagonal element
$\neq 1$ (the case $P={\mathbbm 1}$ can be excluded). Without loss of generality we can
therefore restrict ourselves to the case $P=\mbox{diag}(p,1,1\ldots,1)$.

Let $(\varphi, \psi)$ be the orthonormal basis of ${\mathcal H}_2$, in component representation
\begin{equation}\label{phipsicomp}
  \varphi=(\varphi_1,\varphi_2,\ldots,\varphi_N)\quad \mbox{and }\quad \psi=(\psi_1,\psi_2,\ldots,\psi_N)
  \;.
\end{equation}
Upon multiplication with suitable phase factors
the first components $\varphi_1$ and $\psi_1$ can be chosen as non-negative. The corresponding orthonormal basis
$(\sigma_0(F),\sigma_1(F),\sigma_2(F),\sigma_3(F))$ of $B(F)$ has been defined in (\ref{sigma0}-\ref{sigma3}).
The two-dimensional subspace ${\mathcal H}_2':=P\left({\mathcal H}_2 \right)$ is spanned by
\begin{equation}\label{phip}
\varphi' = P\,\varphi = ( p\,\varphi_1, \varphi_2,\ldots,\varphi_N)
\end{equation}
and
\begin{equation}\label{psip}
\psi' = P\,\psi = ( p\,\psi_1, \psi_2,\ldots,\psi_N)
\;.
\end{equation}
For future reference we compute
\begin{equation}\label{alphap}
 \alpha':= \| \varphi'\|^2 =p^2 |\varphi_1|^2 +  |\varphi_2|^2+ \ldots + |\varphi_N|^2=
 (p^2-1)\varphi_1^2 +  \underbrace{|\varphi_1|^2 +  |\varphi_2|^2+ \ldots + |\varphi_N|^2}_{=\| \varphi\|^2=1}
 =1+ (p^2-1)\varphi_1^2 > 0
 \;,
\end{equation}
and, analogously,
\begin{equation}\label{betap}
 \beta':= \| \psi'\|^2  =1+ (p^2-1)\psi_1^2
 \;,
\end{equation}
using $\varphi_1\ge 0$ and  $\psi_1\ge 0$. Further we calculate
\begin{equation}\label{phippsip}
 \langle \varphi' | \psi'\rangle = p^2\,\varphi_1\,\psi_1+ \sum_{i=2}^{N}\overline{\varphi_i}\,\psi_i
 =(p^2-1)\varphi_1\,\psi_1 +\underbrace{\sum_{i=1}^{N}\overline{\varphi_i}\,\psi_i}_{=\langle \varphi|\psi\rangle=0}
 =(p^2-1)\varphi_1\,\psi_1 \in {\mathbb R}
 \;.
\end{equation}
Since in general $\varphi'$ and $\psi'$ will not be orthogonal
we define
\begin{equation}\label{psipp}
 \psi'':= \psi' - \frac{\langle \varphi' | \psi'\rangle}{\| \varphi'\|^2} \varphi'
 \stackrel{(\ref{alphap},\ref{phippsip})}{=}  \psi' - \frac{(p^2-1)\varphi_1\,\psi_1}{1+ (p^2-1)\varphi_1^2}\, \varphi'
 =: \psi' - \gamma\,\varphi'
 \;,
\end{equation}
such that $\langle \varphi' | \psi''\rangle =0$.
For future reference we calculate
\begin{eqnarray}\label{betapp1}
 \beta''&:=& \| \psi''\|^2 \stackrel{(\ref{psipp})}{=}\langle \psi' - \gamma\,\varphi' | \psi' - \gamma\,\varphi'\rangle
 = \langle \psi' |\psi'\rangle -2 \gamma \langle \varphi' | \psi'\rangle +\gamma^2  \langle \varphi' |\varphi'\rangle\\
 \label{betapp2}
 & \stackrel{(\ref{alphap},\ref{betap},\ref{phippsip})}{=}&
 1+ (p^2-1)\psi_1^2- 2\, \frac{(p^2-1)\varphi_1\,\psi_1}{1+ (p^2-1)\varphi_1^2}\,(p^2-1)\varphi_1\,\psi_1
 +\left( \frac{(p^2-1)\varphi_1\,\psi_1}{1+ (p^2-1)\varphi_1^2}\right)^2\,\left(1+ (p^2-1)\varphi_1^2\right)\\
  \label{betapp3}
 &=& 1+ \frac{(p^2-1)\,\psi_1^2}{1+(p^2-1)\varphi_1^2} >0
 \;,
\end{eqnarray}
after performing some elementary simplifications in the last step of (\ref{betapp3}).

In the two-dimensional subspace ${\mathcal H}_2'$ we will use the orthonormal basis
$(\frac{\varphi'}{\sqrt{\alpha'}},\frac{\psi''}{\sqrt{\beta''}})$.
The corresponding orthonormal standard basis in $B(F')$ assumes the form
\begin{eqnarray}\label{sigmap0}
 \sigma_0(F')&=&\frac{1}{\sqrt{2}}\left(\frac{|\varphi'\rangle \langle \varphi' |}{\alpha'} +
 \frac{|\psi''\rangle \langle \psi'' |}{\beta''}
 \right)=:\tau_0,\\
 \label{sigmap1}
 \sigma_1(F')&=&\frac{1}{\sqrt{2\alpha'\beta''}}\left(|\varphi'\rangle \langle \psi'' | +|\psi''\rangle \langle \varphi'| \right)
 := \frac{1}{\sqrt{\alpha'\beta''}}\,\tau_1,\\
 \label{sigmap2}
 \sigma_2(F')&=&\frac{{\sf i}}{\sqrt{2\alpha'\beta''}}\left(-|\varphi'\rangle \langle \psi'' | + |\psi''\rangle \langle \varphi'| \right)
  := \frac{1}{\sqrt{\alpha'\beta''}}\,\tau_2,\\
 \label{sigmap3}
 \sigma_3(F')&=&\frac{1}{\sqrt{2}}\left(\frac{|\varphi'\rangle \langle \varphi' |}{\alpha'}
  -\frac{|\psi''\rangle \langle \psi'' |}{\beta''} \right)=:\tau_3
  \;,
 \end{eqnarray}
 where we have introduced the operators $\tau_0,\ldots,\tau_3$ for future reference.

Next we calculate the images $\pi_i:=P\,\sigma_i(F)\,P,\,i=0,1,2,3,$ of the standard basis in $B(F)$.
We obtain:
\begin{eqnarray}
\label{pi0a}
 \pi_0 &=& P\,\sigma_0(F)\,P =\frac{1}{\sqrt{2}}\left(|\varphi'\rangle \langle \varphi'| + |\psi'\rangle \langle \psi'|\right)
 = \frac{1}{\sqrt{2}}\left(|\varphi'\rangle \langle \varphi'| + |\psi''+\gamma \varphi'\rangle \langle \psi''+\gamma \varphi'| \right)\\
 \label{pi0b}
 &=& \frac{1}{\sqrt{2}}\left((1+\gamma^2)|\varphi'\rangle \langle \varphi'| + \gamma \left(|\psi''\rangle \langle \varphi'|
 + |\varphi'\rangle \langle \psi''|\right)+|\psi''\rangle \langle \psi''|\right),\\
 \label{pi1a}
 \pi_1 &=& P\,\sigma_1(F)\,P =\frac{1}{\sqrt{2}}\left(|\varphi'\rangle \langle \psi'| + |\psi'\rangle \langle \varphi'|\right)
 = \frac{1}{\sqrt{2}}\left(|\varphi'\rangle \langle \psi''+\gamma \varphi'| + |\psi''+\gamma \varphi'\rangle \langle \varphi'|\right)\\
 \label{pi1b}
  &=& \frac{1}{\sqrt{2}}\left(2\gamma |\varphi'\rangle\langle \varphi'| +|\varphi'\rangle \langle \psi''|+  |\psi''\rangle \langle \varphi'| \right)\;,\\
 \label{pi2a}
 \pi_2 &=& P\,\sigma_2(F)\,P =\frac{\sf i}{\sqrt{2}}\left(-|\varphi'\rangle \langle \psi'| + |\psi'\rangle \langle \varphi'|\right)
 = \frac{\sf i}{\sqrt{2}}\left(-|\varphi'\rangle \langle \psi''+\gamma \varphi'| + |\psi''+\gamma \varphi'\rangle \langle \varphi'|\right)\\
 \label{pi2b}
  &=&\frac{\sf i}{\sqrt{2}}\left(-|\varphi'\rangle \langle \psi''| + |\psi''\rangle \langle \varphi'|\right)\;,\\
  \label{pi3a}
 \pi_3 &=& P\,\sigma_3(F)\,P =\frac{1}{\sqrt{2}}\left(|\varphi'\rangle \langle \varphi'| - |\psi'\rangle \langle \psi'|\right)
 = \frac{1}{\sqrt{2}}\left(|\varphi'\rangle \langle \varphi'| - |\psi''+\gamma \varphi'\rangle \langle \psi''+\gamma \varphi'| \right)\\
 \label{pi3b}
 &=& \frac{1}{\sqrt{2}}\left((1-\gamma^2)|\varphi'\rangle \langle \varphi'| - \gamma \left(|\psi''\rangle \langle \varphi'|
 + |\varphi'\rangle \langle \psi''|\right)-|\psi''\rangle \langle \psi''|\right)\;.
 \end{eqnarray}

According to Proposition \ref{Tdet} we will have to show that $\det {\sf T}(F)\ge 0$, where ${\sf T}(F)$
is the matrix with entries
\begin{equation}\label{Tij}
 {\sf T}_{ij}=\mbox{Tr} \left(\pi_j \, \sigma_i(F') \right)
 \;,
\end{equation}
since the expansion coefficients are given by scalar products in $B({\mathcal H})$.
It turns out to be easier to use the related matrix $\Lambda$ with entries
\begin{equation}\label{Lij}
 \Lambda_{ij}=\mbox{Tr} \left(\pi_j \, \tau_i \right)
 \;,
\end{equation}
where the $\tau_i$ are defined in (\ref{sigmap0}) - (\ref{sigmap3}). It will suffice to prove
$\det \Lambda\ge 0$ since $\det  \Lambda= \alpha'\,\beta''\,\det {\sf T}(F)$ and $\alpha', \beta''>0$.\\

We provide a numerical example for the matrix $\Lambda$ calculated for $N=4$:
\begin{equation}\label{Lnum}
  \Lambda=\left(
\begin{array}{cccc}
 0.862251 & 0.00428142 & 0 & 0.137683 \\
 0.0031021 & 0.724502 & 0 & -0.0031021 \\
 0 & 0 & 0.724502 & 0 \\
 0.137701 & 0.00428142 & 0 & 0.862233 \\
\end{array}
\right)
\;.
\end{equation}
This example suggests that the following identities hold for the $\Lambda_{ij}$:
%%%%%%%%%%%%%%%%%%%%%%%%%%%%%%%%%%%%%%%%%%%%%%%%%%%%%%%%%%%%%%%%%%%%%%%%%%%%%%%%%%%%%%%%%%%%%%%%%%%
\begin{prop}\label{Pid}
 \begin{eqnarray}
 \label{Pid1}
   \Lambda_{02} &=& \Lambda_{12}= \Lambda_{32}=0\;, \\
   \label{Pid2}
   \Lambda_{20} &=& \Lambda_{21}= \Lambda_{23}=0\;, \\
   \label{Pid3}
   \Lambda_{10}+\Lambda_{13}&=& 0\;, \\
    \label{Pid4}
   \Lambda_{01}&=& \Lambda_{31}\;, \\
    \label{Pid5}
   \Lambda_{11}&=& \Lambda_{22}\;, \\
   \label{Pid6}
   \Lambda_{00}+ \Lambda_{03}&=& \ \Lambda_{30}+ \Lambda_{33}\;.
 \end{eqnarray}
\end{prop}
%%%%%%%%%%%%%%%%%%%%%%%%%%%%%%%%%%%%%%%%%%%%%%%%%%%%%%%%%%%%%%%%%%%%%%%%%%%%%%%%%%%%%%%%%%%%%%%%%%
The proof of Proposition \ref{Pid} is quite long and will be therefore moved to the Appendix \ref{APid}.
Due to the identities (\ref{Pid1}) - (\ref{Pid6}) the eigenvalues $\lambda_i,\; i=0,1,2,3,$ of $\Lambda$ can be
explicitly calculated with the result:
\begin{eqnarray}\label{evL0}
\lambda_0&=& \Lambda_{00}+\Lambda_{03} = \alpha' >0\;,\\
\label{evL1}
\lambda_1&=& \Lambda_{00}-\Lambda_{30} = \beta'' >0\;,\\
\label{evL23}
\lambda_2&=&  \lambda_3 =\Lambda_{22}= \alpha'\,\beta'' >0\;.
\end{eqnarray}

This immediately implies $\det \Lambda>0$ and the proof for the case (i) is complete.
Generally, we have only shown that $\det \Lambda \ge 0$ in the case (i),
since we have also taken into account some degenerate cases where $\det \Lambda = 0$.

\subsubsection{Case (iii) } \label{Ciii}

In this case the pure operation $T$ is of the form $T(X)= \mbox{Tr}(W X)\,P_\phi$ and completely positive.
The quadruple $(T\, X_0, T\, X_1, T\, X_2, T\, X_3)$ will be always linearly dependent
and hence $\mbox{Det}(T\, X_0, T\, X_1, T\, X_2, T\, X_3)=0$, see Corollary \ref{C1}.
Therefore $T$ is orientation-preserving.

\subsubsection{Case (ii) } \label{Cii}

In this case there exists a complex anti-linear operator $B: {\mathcal H} \to {\mathcal H}$ such that
$T(X)=B X B^\ast$ for all $X\in B({\mathcal H})$. $B$ can be written as a product $B=A\,C$, where
$A$ is a linear operator and $C$ a special anti-unitary one. As we have seen, in this case both possibilities
can occur: namely, $T$ may be completely positive or it may not be.

As in the case (i) we consider an
arbitrary three-dimensional face $F\in{\mathcal F}_3$ of operators with support in ${\mathcal H}_2$
and an orthonormal basis $(\varphi,\psi)$ in ${\mathcal H}_2$.
As the special anti-unitary operator $C$ we choose an extension of the complex-conjugation in ${\mathcal H}_2$
w.~r.~t.~the basis $(\varphi,\psi)$.  Otherwise, we proceed exactly as in case (i). The only difference
is that the anti-unitary operator $C$ operates on the standard basis of $B(H)$ in the following way:
\begin{equation}\label{Cop}
  C\,\sigma_i(F)\,C=\sigma_i(F)\quad\mbox{for } i=0,1,3,\quad \mbox{but}\quad  C\,\sigma_2(F)\,C=-\sigma_2(F)
  \;.
\end{equation}
This implies that the third column of the matrix $\Lambda$ in the case (i) will be replaced by its negative,
i.e, $\Lambda_{22}\mapsto -\Lambda_{22}$. Consequently, $\det \Lambda \le 0$.

We proceed by the following Proposition which will be proven in the Appendix \ref{ACPP}:
%%%%%%%%%%%%%%%%%%%%%%%%%%%%%%%%%%%%%%%%%%%%%%%%%%%%%%%%%%%%%%%%%%%%%%%%%%%%%%%%%%%%%%%%%%%%%%%%%
\begin{prop}\label{CPP}
 Let $T:B({\mathcal H})\to B({\mathcal H})$ be a pure operation.
 Then $T$ is completely positive iff it is of type (i) or (iii).
\end{prop}
%%%%%%%%%%%%%%%%%%%%%%%%%%%%%%%%%%%%%%%%%%%%%%%%%%%%%%%%%%%%%%%%%%%%%%%%%%%%%%%%%%%%%%%%%%%%%%%%%%%

Let us first consider the case where $\det \Lambda < 0$ holds for some $F\in{\mathcal F}_3$.
Then $T$ will not be orientation-preserving.
In this case, $T$ is also not completely positive: Otherwise,
according to Proposition \ref{CPP}, $T$ would fall under case (i) or (iii) of Theorem \ref{TPO},
and thus $\det \Lambda \ge 0$ for all $F\in{\mathcal F}_3$, as we have shown above.
This contradicts the assumption that $\det \Lambda < 0$ for some $F\in{\mathcal F}_3$.

We are thus left with the case  $\det \Lambda =0$ for all $F\in{\mathcal F}_3$.
It follows that $A$ will map each two-dimensional subspace ${\mathcal H}_2$ corresponding to $F$ onto some
one-dimensional subspace $S(F)$ or onto $\{0\}$. As in the proof of Proposition \ref{CPP}, case B,
one can show that all one-dimensional subspaces $S(F)$ must coincide and hence we are in the case (iii).
$T$ is thus completely positive and orientation-preserving.      \hfill$\square$

\subsubsection{Outlook for the case of non-pure operations} \label{OUT}

The question arises as to whether the characterization of completely positive pure operations given
in Theorem \ref{Tpure} can also be applied to general operations that are not necessarily pure.
Recall that the operator sum representation (\ref{Kraus}) shows that every completely positive operation
can be written as the sum of pure ones. However, counter examples show that a completely positive operation
need not be orientation-preserving. The reason for this can be traced back to the fact that the sum
of, say, two matrices with positive determinants can have a negative determinant. A simple example is
\begin{equation}\label{counterdet}
 X=\left( \begin{array}{cc}
            1 & 1/2 \\
            1/2 & 1
          \end{array}\right),\quad
 Y=\left( \begin{array}{cc}
            -1 & 1/2 \\
            1/2 & -1
          \end{array}\right),\quad
 X+Y=\left( \begin{array}{cc}
           0 & 1 \\
            1 &0
          \end{array}\right)\;,
\end{equation}
satisfying $\det X=\det Y=3/4$ and $\det (X+Y)=-1$.
Therefore, characterizing completely positive operations by respecting the
Noether structure in any sense will probably take a more complicated form.

\section{Summary}\label{sec:SUM}

This paper outlines a step-by-step reconstruction of QT that is
closely based on physically interpretable structures.
We start from the statistical structure of quantum mechanics,
as described most extensively by G.~Ludwig. Then, following the work of Alfsen and Shultz,
we propose additionally considering the ``quantum orientation" or, equivalently,  ``Noether structure,"
which essentially mediates an isomorphism between the space of (first moments of) observables
and the Lie algebra of the automorphism group of a physical system.
If one assumes physical units for the spaces mentioned,
the Noether structure contains a conversion factor with the dimension of ``inverse action",
which can naturally be identified with $\hbar^{-1}$.

In our view, the Noether structure is a prerequisite for the
definition of composite systems and completely positive state transformations.
We have provided arguments to support this: on the one hand, we have shown that the
composition of systems is not uniquely determined by using only the statistical structure;
on the other hand, in the finite-dimensional case, pure completely positive operations
can be characterized by the property of preserving the quantum orientation.
However, given the relatively complex proof and the restriction to pure operations,
it is not yet advisable to include orientation preservation in addition to the three known
equivalent textbook definitions of completely positive operations \cite{NC00}.
This will likely require further work.

\section*{Acknowledment}
This paper is dedicated to J\"urgen Schnack on the occasion of his 60th birthday.
Working with him has always been a pleasure, and I have learned a great deal from him,
which has helped me find my orientation in the world of quantum physics.
I would also like to thank Pekka Lahti for our discussions on the topic of
this article and for providing some references.

\appendix

\begin{figure}[ht]
  \centering
    \includegraphics[width=0.7\linewidth]{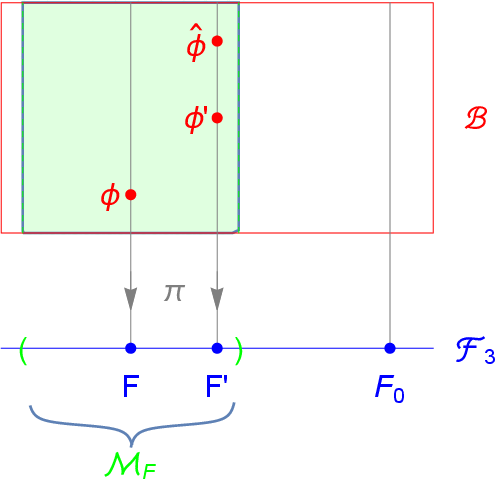}
  \caption{
   Sketch of the principal $O(3)$-bundle $\pi: {\mathcal B}\to{\mathcal F}_3$ and some points
   that appear in subsection \ref{LT}. The light green area represents the inverse image
   $\pi^{-1} ({\mathcal M}_F)$ of the neighbourhood ${\mathcal M}_F$ of $F$.
   }
  \label{FIGPFB}
\end{figure}

%\newpage

\section{Notes on quantum orientation}\label{NQO}

\subsection{Local triviality of the  $O(3)$-bundle $\pi: {\mathcal B}\to {\mathcal F}_3$}\label{LT}

We will sketch the argument for the required local triviality of the $O(3)$-bundle $\pi: {\mathcal B}\to {\mathcal F}_3$,
see Fig.~ \ref{FIGPFB}.
Let $F\in {\mathcal F}_3$ and ${\mathcal M}_F$ be some open neighbourhood of $F$.
Consider the principal $U(2)\times U(N-2)$-bundle $\Psi: U(N)\to U(N)/(U(2)\times U(N-2))\cong {\mathcal F}_3$
and a local section $u: {\mathcal M}_F  \to U({\mathcal H}) \cong  U(N)$.

We have to show that there exists a homeomorphism
\begin{equation}\label{defGamma}
 \Gamma: \pi^{-1} \left({\mathcal M}_F\right) \to {\mathcal M}_F\times O(3)
 \;,
\end{equation}
such that $\Gamma(\psi)=\left(\pi(\psi), R(\psi)\right)$ for all $\psi\in \pi^{-1} \left({\mathcal M}_F\right)$.
Let $\phi'\in  \pi^{-1} \left({\mathcal M}_F\right)$ such that $\phi': E^3 \to F'\in {\mathcal M}_F$
and fix some ${\phi}\in \pi^{-1} (F)$ hence satisfying ${\phi}: E^3 \to F$.
Recall the fixed two-dimensional subspace ${\mathcal H}_2^0$ and denote by $F_0\in {\mathcal F}_3$ the corresponding face.
$F$ and $F'$ are unitary transforms of $F_0$; for the corresponding unitary operators we use the local section
$u$ considered above. We thus define $S_\phi: F_0 \to F$ by $S_\phi(f_0) = u(\phi)\,f_0\,u(\phi)^\ast$  for all $f_0\in F_0$, and,
analogously, $S_{{\phi}'}: F_0 \to F'$ by $S_{{\phi}'}(f_0) = u({{\phi'}})\,f_0\,u({{\phi}'})^\ast$.
We further consider the affine isomorphism $\hat{\phi}: E_3 \to F'$ defined by the composition
$\hat{\phi} :
E^3\stackrel{\phi}{\rightarrow} F
\stackrel{S_{\phi}^{-1}}{\rightarrow}F_0
\stackrel{S_{\phi'}}{\rightarrow}F'$.
It follows that $R(\phi'):=\hat{\phi}^{-1}\,{\phi'}$ will be an affine isomorphism of $E^3$ onto itself, and hence a rotation $R(\phi')\in O(3)$.
We will define $\Gamma(\phi') =\left( F', R(\phi')\right)$
and omit the proof that $\Gamma: \pi^{-1} \left({\mathcal M}_F\right) \to {\mathcal M}_F\times O(3)$
will be a homeomorphism.

\subsection{Global sections of the  ${\mathbb Z}_2$-bundle $\Pi: {\mathcal B}/SO(3)\to {\mathcal F}_3$}\label{GS}

We will define two distinct global sections $\Sigma^\pm$ of the
${\mathbb Z}_2$-bundle $\Pi: {\mathcal B}/SO(3)\to {\mathcal F}_3$.
Let us consider an arbitrary $3$-dimensional face $F\in{\mathcal F}_3$ corresponding to
a two-dimensional subspace ${\mathcal H}_2$. Let $(\varphi,\psi)$ be any orthonormal
basis in ${\mathcal H}_2$, and $B(F)\cong B({\mathcal H}_2)$ be the $4$-dimensional
subspace of $B({\mathcal H})$ spanned by all statistical operators $W\in F$ or,
equivalently, consisting of all Hermitean operators with support included in ${\mathcal H}_2$.
We define the following two orthonormal bases in $B(F)$:
\begin{eqnarray}\label{sigmapm0}
 \sigma^\pm_0(F)&=&\frac{1}{\sqrt{2}}\left(|\varphi\rangle \langle \varphi | +|\psi\rangle \langle \psi | \right),\\
 \label{sigmapm1}
 \sigma^\pm_1(F)&=&\frac{1}{\sqrt{2}}\left(|\varphi\rangle \langle \psi | +|\psi\rangle \langle \varphi| \right),\\
 \label{sigmapm2}
 \sigma^\pm_2(F)&=&\pm\frac{{\sf i}}{\sqrt{2}}\left(-|\varphi\rangle \langle \psi | +|\psi\rangle \langle \varphi| \right),\\
 \label{sigmapm3}
 \sigma^\pm_3(F)&=&\frac{1}{\sqrt{2}}\left(|\varphi\rangle \langle \varphi | -|\psi\rangle \langle \psi | \right).
 \end{eqnarray}
$(\sigma_0^+(F))_{i=0,1,2,3}$  will be referred to as the ``standard basis" in $B(F)$.
This gives rise to two
affine isomorphisms $\phi^\pm:E^3 \to F$ in the following way:
\begin{eqnarray}
\label{defphi1}
  \tau^\pm_0(F)=\frac{1}{\sqrt{2}} \, \sigma^\pm_0(F),\; \tau^\pm_i(F)&=&
  \tau^\pm_0(F) +\frac{1}{\sqrt{2}} \,\sigma^\pm_i(F),\quad \mbox{for}\; i=1,2,3, \\
  \label{defphi2}
  \phi^\pm(F)(0,0,0) &=& \tau^\pm_0(F),\\
   \label{defphi3}
  \phi^\pm(F)(1,0,0) &=& \tau^\pm_1(F),\\
   \label{defphi4}
  \phi^\pm(F)(0,1,0) &=& \tau^\pm_2(F),\\
   \label{defphi5}
  \phi^\pm(F)(0,0,1) &=& \tau^\pm_3(F)
  \;.
\end{eqnarray}
Here we have used the identification of $E^3$ with the unit ball in ${\mathbb R}^3$ and that an affine
isomorphism $\phi:E^3 \to F$ can be uniquely determined by the image of four point that affinely generate ${\mathbb R}^3$.

Then we define $\Sigma^\pm: {\mathcal F}_3\to {\mathcal B}/SO(3)$ by $\Sigma^\pm(F):=\phi\pm(F)\, SO(3)$ for all $F\in  {\mathcal F}_3$.
Clearly, these are two different cross sections of the  ${\mathbb Z}_2$-bundle $\Pi: {\mathcal B}/SO(3)\to {\mathcal F}_3$.
It remains to show that the $\Sigma^\pm:  {\mathcal F}_3 \to  {\mathcal B}/SO(3)$ are continuous maps.
Consider a fixed face $F\in {\mathcal F}_3$ and
let $\sigma: {\mathcal M}_F \to {\mathcal B}$ be a local section of the  $O(3)$-bundle $\pi: {\mathcal B}\to {\mathcal F}_3$
satisfying $\sigma(F) = \phi^+(F)$ according to the definition (\ref{defphi1} - \ref{defphi5}). Consider an arbitrary face
$F'\in  {\mathcal M}_F$ and let $\phi'=\phi^+(F')$ such that $\phi':E^3\to F'$,
moreover $\phi''=\sigma(F')$ such that also  $\phi':E^3\to F'$. It follows that $\phi'' = \phi'\, R$ for
some $R\in O(3)$ since $\phi'$ and $\phi''$ lie in the same fiber over $F'$. By continuity of $\sigma$
we have $\det R=1$, i.~e.,  $R\in SO(3)$. Let $\Phi: {\mathcal B}\to {\mathcal B}/SO(3)$ be the continuous quotient map
that assigns to each $\phi\in {\mathcal B}$ the orbit $\Phi(\phi)=\phi\,SO(3)$, then we have shown that
$\Sigma^+$ and $\Phi \circ \sigma$ coincide on  ${\mathcal M}_F$.
Therefore, $\Sigma^+$ is continuous on ${\mathcal M}_F$ and thus also on the entire ${\mathcal F}_3$,
which can be covered by open neighborhoods of the form ${\mathcal M}_F$.
The proof of the continuity of $\Sigma^-$ proceeds in exactly the same way.

\section{Proof of Proposition \ref{P2}}\label{AP2}

Let $I$ denote the multiplication with the imaginary unit ${\sf i}$ and write $\psi= I \circ\Psi$
such that $\Psi: B({\mathcal H}) \to B_0({\mathcal H})$ is a linear map.
First, we only consider the conditions  (\ref{cond1}) and  (\ref{cond2}) of Proposition \ref{P2}.
From (\ref{cond2})
it follows that $[\Psi(a),a]=0$ for all $a\in B({\mathcal H})$. Choose some  $a\in B_0({\mathcal H})$, i.~e.,
$a\in B({\mathcal H})$ with $\mbox{Tr } a=0$ and  a spectral decomposition
$a=\sum_{n=1}^{N}a_n P_n$ where all $a_n$ are pairwise different, $a_1=1,a_2=-1$
and the $P_n$ are one-dimensional projectors satisfying $\sum_{n=1}^{N}P_n={\mathbbm 1}$.
Since $\Psi(a)\in B_0({\mathcal H})$ commutes with $a$ it must be of the form $\Psi(a)=\sum_{n=1}^{N}b_n P_n$ such that $\sum_{n=1}^{N}b_n=0$.
Let $a_n \to 0$ for $n=3,\ldots,N$. Then $\Psi(a)\to \Psi(P_1-P_2)=\sum_{n=1}^{N}\tilde{b}_n P_n$. Since
the projectors $P_3, \ldots, P_N$ can be chosen as forming an arbitrary system of orthogonal projectors in the null space
of $P_1+P_2$ the coefficients $\tilde{b}_n$ must be constant for $n=3,\ldots,N$ and hence
\begin{eqnarray}
\label{PsiP1P2a}
   \Psi(P_1-P_2) &=& \tilde{b}_1 P_1+ \tilde{b}_2 P_2+\tilde{b}_3 \sum_{n=3}^{N}P_n\\
   \label{PsiP1P2b}
  &=&  \tilde{b}_1 P_1+ \tilde{b}_2 P_2+N \tilde{b}_3 \left( \frac{1}{N}{\mathbbm 1} - \frac{1}{N}(P_1+P_2)\right) \\
  \label{PsiP1P2c}
  &=&\underbrace{ \left(\tilde{b}_1 +N \tilde{b}_3\left(\frac{1}{2}-\frac{1}{N} \right) \right)}_{:=\lambda_1} P_1 +
 \underbrace{ \left(\tilde{b}_2 +N \tilde{b}_3\left(\frac{1}{2}-\frac{1}{N} \right) \right)}_{:=\lambda_2} P_2 +
 \underbrace{N \tilde{b}_3}_{:=\mu}
   \underbrace{\left(\frac{1}{N}{\mathbbm 1}-\frac{1}{2}(P_1+P_2)\right)}_{:=\sigma_{12}\in B_0({\mathcal H})}\\
   \label{PsiP1P2d}
  &=&\lambda\,(P_1-P_2)+ \mu\,\sigma_{12}
 \;,
\end{eqnarray}
where we have used in (\ref{PsiP1P2d}) that $\mbox{Tr }\left( \Psi(P_1-P_2) \right)=0$ and hence  $\lambda_1=-\lambda_2=:\lambda$.
Note that for $N=2$ we have $\sigma_{12}=0$. To exclude this special case, we will provisionally assume that $N>2$ in the following.

Let ${\mathcal F}$ be the face of $K(\mathcal H)$ generated by $P_1,P_2$, affinely isomorphic to a $3$-ball,
and  ${\mathcal S}_0:= \{ P-P'\left|\right. P,P'\in {\mathcal F} \mbox{ orthogonal one-dimensional projectors}\}$.
${\mathcal S}_0$ is linearly and isometrically isomorphic to a $3$-sphere with radius $\sqrt{2}$.
Moreover, let ${\mathcal R}_0$ denote the $4$-dimensional Euclidean subspace of $B_0({\mathcal H})$ spanned by ${\mathcal S}_0$
and $\sigma_{12}=\left(\frac{1}{N}{\mathbbm 1}-\frac{1}{2}(P_1+P_2)\right)$. Equation (\ref{PsiP1P2d})
can be rewritten by setting $P_1=P$ and $P_2=P'$ and then shows that $\Psi$ maps ${\mathcal S}_0$ injectively into ${\mathcal R}_0$:
\begin{equation}\label{Psie}
  \Psi\left( \sqrt{2}\,{\mathbf e}\right) = \lambda({\mathbf e})\, \sqrt{2}\, {\mathbf e} + \mu({\mathbf e}) \sigma_{12}
  \;.
\end{equation}
This holds for arbitrary $ \sqrt{2}\,{\mathbf e}\in {\mathcal S}_0$ such that $\| {\mathbf e} \|=1$.
Here we have explicitly indicated the possible ${\mathbf e}$-dependence of $\lambda$ and $\mu$.
Upon introducing suitable Euclidean coordinates in
${\mathcal R}_0$ we may identify the restriction of $\Psi$  to ${\mathcal S}_0$
with a real $4\times 3$-matrix with entries $\Psi_{ij}$.
W.~r.~t.~these coordinates $\sigma_{12}$ assumes the form $\sigma_{12}=(0,0,0,s_{12})$
and ${\mathbf e}=({e}_1,{e}_2,{e}_3)$.

We will choose ${\mathbf e}=(1,0,0)$. Then (\ref{Psie}) implies that $\Psi_{11}=\lambda(1,0,0)$ and
$\Psi_{21}=\Psi_{31}=0$. Analogously for ${\mathbf e}=(0,1,0)$ and ${\mathbf e}=(0,0,1)$.
Hence the $3\times 3$-submatrix of $\Psi$ with entries $\Psi_{ij},\;i,j=1,2,3$ must be of diagonal form.
Moreover,
\begin{equation}\label{Psidiag}
  \left(
  \begin{array}{ccc}
    \Psi_{11} & 0 & 0 \\
    0 & \Psi_{22} & 0 \\
    0 & 0 & \Psi_{33}
  \end{array}
  \right)\;\mathbf{e} = \lambda({\mathbf e})\,{\mathbf e}
\end{equation}
shows that all diagonal elements $\Psi_{ii},\,i=1,2,3$ have the constant value $\lambda= \lambda\,({\mathbf e})$
which is hence independent of ${\mathbf e}$. The remaining equation corresponding to the $4$th component of (\ref{Psie}) reads
\begin{equation}\label{Psi4}
  \sqrt{2} \left(\Psi_{41} {e}_1+\Psi_{42} {e}_2+\Psi_{43} {e}_3\right)=
  \mu({\mathbf e}) s_{12}
  \;.
\end{equation}

Our next aim is to show that  $\mu({\mathbf e})=0$. For this we provisionally assume $N>3$.
Again, let $\left(P_i\right)_{i=1,\ldots,N}$ be a general system of mutually orthogonal one-dimensional projectors.
By adding the two equations
\begin{eqnarray}
\label{Psi12}
  \Psi(P_1-P_2) &=& \lambda_{12}\left(P_1-P_2 \right) +\mu_{12}\left(\frac{1}{N}{\mathbbm 1}-\frac{1}{2}(P_1+P_2)\right)\;,\\
  \label{Psi23}
 \Psi(P_2-P_3) &=& \lambda_{23}\left(P_2-P_3 \right) +\mu_{23}\left(\frac{1}{N}{\mathbbm 1}-\frac{1}{2}(P_2+P_3)\right)\;,
\end{eqnarray}
corresponding to (\ref{PsiP1P2d}) we obtain
\begin{eqnarray}
\label{Psi13a}
  \Psi(P_1-P_3) &=& \left(\lambda_{12}-\frac{1}{2}\mu_{12} \right) P_1+
  \left(\lambda_{23}-\lambda_{12}-\frac{1}{2}(\mu_{12}+\mu_{23}) \right)P_2
  +\left( -\lambda_{23}-\frac{1}{2}\mu_{23}\right) P_3
  +\frac{1}{N}\left( \mu_{12}+\mu_{23}\right) {\mathbbm 1}
   \\
   \label{Psi13b}
   &=& \lambda_{13}\left(P_1-P_3 \right) +\mu_{13}\left(\frac{1}{N}{\mathbbm 1}-\frac{1}{2}(P_1+P_3)\right)
   \;.
\end{eqnarray}
Due to $N>3$ the four operators $P_1,P_2,P_3,{\mathbbm 1}$ are linearly independent in $B({\mathcal H})$.
Hence the identity of (\ref{Psi13a}) and (\ref{Psi13b}) entails a system of four linear equations
for the six unknowns $\lambda_{12},\ldots, \mu_{13}$.
It solution reads
\begin{eqnarray}
\label{mu12}
  \mu_{12} &=&2(\lambda_{23}-\lambda_{13})\;, \\
  \label{mu23}
  \mu_{23} &=& 2(\lambda_{13}-\lambda_{12})\;, \\
  \label{mu13}
  \mu_{13}&=&  2(\lambda_{23}-\lambda_{12})
  \;.
\end{eqnarray}
It follows that not only the $\lambda_{ij}$, but also the $\mu_{ij}$ are independent of ${\mathbf e}$.
For the case $N=3$ the proof is analogous, using that $P_1, P_2, P_3$ are linearly independent.

Reconsidering eq.~(\ref{Psi4}) which now assumes the form
\begin{equation}\label{Psi4a}
  \sqrt{2} \left(\Psi_{41} {e}_1+\Psi_{42} {e}_2+\Psi_{43} {e}_3\right)=
  \mu_{12} s_{12}
  \;,
\end{equation}
where $\mu_{12}$ is independent of ${\mathbf e}$, we  conclude that it cannot be satisfied for all
${\mathbf e}$ with $\|{\mathbf e}\|=1 $, unless $\mu_{12}=\Psi_{41}=\Psi_{42}=\Psi_{43}=0$.
(In the case $N=2$ this trivially holds because $s_{12}=0$.)
Inserting this result into (\ref{mu12}) gives $\lambda_{23}=\lambda_{13}$.
Slightly generalizing this result, we may state that $\lambda_{ij}=\lambda_{k\ell}$ for all pairs of faces
${\mathcal F}_{ij},\, {\mathcal F}_{k\ell}$ generated by two orthogonal projectors $P_i,P_j$ , resp.~$P_k,P_\ell$,
of the family $\left(P_i\right)_{i=1,\ldots,N}$.

Thus we have shown that $\Psi$ acts as the multiplication with a factor $\lambda_{{\mathcal F}}$ on each sphere ${\mathcal S}_0$
formed by differences $P-P'$ of any two orthogonal projectors from a face ${\mathcal F}$  of
$K({\mathcal H})$ generated by two orthogonal one-dimensional projectors $P_1,P_2$.
It remains to show that $\lambda_{{\mathcal F}}$ is constant, i.~e., does not depend on ${\mathcal F}$.
The above result $\lambda_{ij}=\lambda_{k\ell}$ is a first step of this proof.
The claim follows since there exists a basis, albeit not an orthonormal one, in $B_0({\mathcal H})$
formed of differences $P-P'$ of orthogonal projectors from faces of the form ${\mathcal F}_{ij}$.
which will be described in the following.
Let $(|i\rangle)_{i=1,\ldots ,N}$ denote an orthonormal basis in ${\mathcal H}$ such that $P_i=|i\rangle \langle i|$
for $i=1,\ldots,N$. Then one part of the basis in $B_0({\mathcal H})$ is formed by the ``generalized Pauli matrices"
$|k\rangle \langle \ell|+|\ell\rangle \langle k|$ and ${\sf i}(|k\rangle \langle \ell|-|\ell\rangle \langle k|)$, where
$1\le k < \ell \le N$. The remaining $N-1$ basis matrices can be chosen as $P_1-P_k, \, k=2,\ldots,N$.
Since $\Psi$ is multiplication with a constant $\lambda$ on this basis it will be of this form for all $B_0({\mathcal H})$.
Further, $\lambda \neq 0$ follows from (\ref{cond1}).

Finally we invoke condition  (\ref{cond3}) and conclude that $\lambda =\pm 1$,
by which the proof of proposition \ref{P2} is complete.

\section{Proof of Proposition \ref{Pid}}\label{APid}

\subsection{Proof of (\ref{Pid1})}\label{AP1}

These identities follow from the following Lemma:
\begin{lemma}\label{LAP1}
 $ \pi_2=\tau_2$ .
\end{lemma}
If this Lemma holds, then the expansion coefficients of $\pi_2=\sum_{i=0}^{3} \Lambda_{i2}\,\tau_i$
vanish except $\Lambda_{22}$.\\
{\bf Proof}:
\begin{equation}\label{AP1a}
 \pi_2 \stackrel{(\ref{pi2b})}{=}\frac{\sf i}{\sqrt{2}}\left(-|\varphi'\rangle \langle \psi''| + |\psi''\rangle \langle \varphi'|\right)
 \stackrel{(\ref{sigmap2})}{=} \tau_2
 \;.  \quad\quad\quad\quad \square
\end{equation}

In the following proofs we have to evaluate some entries $\Lambda_{ij}=\mbox{Tr} \left(\pi_j\,\tau_i \right)$.
The involved operators have been expressed as linear combinations of the four operators
\begin{equation}\label{op4}
 |\varphi'\rangle\langle \varphi'|,\quad
 |\varphi'\rangle\langle \psi''|,\quad
 |\psi''\rangle\langle \varphi'|,\quad
 |\psi''\rangle\langle \psi''|\;.
\end{equation}
It can be easily shown that only four of the possible products of (\ref{op4}) have a non-vanishing trace, namely
\begin{eqnarray}\label{possprod1}
&& \mbox{Tr} \left( |\varphi'\rangle\langle \varphi'| \varphi'\rangle\langle \varphi'| \right)=\alpha'^2\\
\label{possprod2}
&& \mbox{Tr} \left( |\varphi'\rangle\langle \psi''| \psi''\rangle\langle \varphi'| \right)=\alpha'\,\beta''\\
\label{possprod3}
&& \mbox{Tr} \left( |\psi''\rangle\langle \varphi'| \varphi'\rangle\langle \psi''| \right)=\alpha'\,\beta''\\
\label{possprod4}
&& \mbox{Tr} \left( |\psi''\rangle\langle \psi''| \psi''\rangle\langle \psi''| \right)=\beta''^2\;,
\end{eqnarray}
which facilitates the following proofs.

\subsection{Proof of (\ref{Pid2})}\label{AP2}

\subsubsection{$\Lambda_{20}=0$}

Taking into account the products (\ref{possprod2}) and (\ref{possprod3}) with non-vanishing trace
we obtain from (\ref{pi0b}) and (\ref{sigmap2}):
\begin{equation}
\label{Pid21}
  \Lambda_{20} = \mbox{Tr}\,\pi_0\,\tau_2=
  \frac{{\sf i}\,\gamma}{2}  \mbox{Tr}\,\left( -|\varphi'\rangle\langle \psi''| \psi''\rangle\langle \varphi'|
  +|\psi''\rangle\langle \varphi'| \varphi'\rangle\langle \psi''|\right)=0\;.
 \end{equation}

\subsubsection{$\Lambda_{21}=0$}

Taking into account the products (\ref{possprod2}) and (\ref{possprod3}) with non-vanishing trace
we obtain from (\ref{pi1b}) and (\ref{sigmap2}):
\begin{equation}
\label{Pid22}
  \Lambda_{21} = \mbox{Tr}\,\pi_1\,\tau_2=
  \frac{{\sf i}}{2}  \mbox{Tr}\,\left( -|\varphi'\rangle\langle \psi''| \psi''\rangle\langle \varphi'|
  +|\psi''\rangle\langle \varphi'| \varphi'\rangle\langle \psi''|\right)=0\;.
 \end{equation}

\subsubsection{$\Lambda_{23}=0$}

Taking into account the products (\ref{possprod2}) and (\ref{possprod3}) with non-vanishing trace
we obtain from (\ref{pi3b}) and (\ref{sigmap2}):
\begin{equation}
\label{Pid23}
  \Lambda_{23} = \mbox{Tr}\,\pi_3\,\tau_2=
  -\frac{{\sf i}\,\gamma}{2}  \mbox{Tr}\,\left( -|\varphi'\rangle\langle \psi''| \psi''\rangle\langle \varphi'|
  +|\psi''\rangle\langle \varphi'| \varphi'\rangle\langle \psi''|\right)=0\;.
 \end{equation}

\subsection{Proof of (\ref{Pid3})}\label{AP3}

Note that $\Lambda_{10}+\Lambda_{13}= \mbox{Tr}\,(\pi_0+\pi_3)\,\tau_1$
and, due to (\ref{pi0b}) and (\ref{pi3b}), $\pi_0+\pi_3=\sqrt{2}|\varphi'\rangle\langle \varphi'|$.
From  (\ref{sigmap1}) it follows that all possible products with the terms of  $\tau_1$
have vanishing trace and hence $\Lambda_{10}+\Lambda_{13}=0$.

\subsection{Proof of (\ref{Pid4})}\label{AP4}

Note that $\Lambda_{01}-\Lambda_{31}= \mbox{Tr}\,\pi_1\,(\tau_0-\tau_3)$
and, due to (\ref{sigmap0}) and (\ref{sigmap3}), $\tau_0-\tau_3=\frac{\sqrt{2}}{\beta''}|\psi''\rangle\langle \psi''|$.
From  (\ref{pi1b}) it follows that all possible products with the terms of  $\pi_1$
have vanishing trace and hence $\Lambda_{01}-\Lambda_{31}=0$.

\subsection{Proof of (\ref{Pid5})}\label{AP5}

First, consider
\begin{equation}\label{Pid51}
  \Lambda_{11}= \mbox{Tr}\,\pi_1\,\tau_1 = \frac{1}{2} \mbox{Tr}\left(  |\varphi'\rangle\langle \psi''| \psi''\rangle\langle \varphi'|
  + |\psi''\rangle\langle \varphi'| \varphi'\rangle\langle \psi''|\right)
  \stackrel{(\ref{possprod2},\ref{possprod3})}{=}\alpha'\,\beta''
  \;.
\end{equation}
Analogously,
\begin{equation}\label{Pid52}
  \Lambda_{22}= \mbox{Tr}\,\pi_2\,\tau_2 = \frac{1}{2} \mbox{Tr}\left(  |\varphi'\rangle\langle \psi''| \psi''\rangle\langle \varphi'|
  + |\psi''\rangle\langle \varphi'| \varphi'\rangle\langle \psi''|\right)
  \stackrel{(\ref{possprod2},\ref{possprod3})}{=}\alpha'\,\beta''
  \;.
\end{equation}
This proves $ \Lambda_{11}= \Lambda_{22}$.

\subsection{Proof of (\ref{Pid6})}\label{AP6}

Note that $\Lambda_{30}+\Lambda_{33}= \mbox{Tr}\,\tau_3\,(\pi_0+\pi_3)$
and, due to  (\ref{pi0b}) and  (\ref{pi3b}), $\pi_0+\pi_3=\sqrt{2}|\varphi'\rangle\langle \varphi'|$.
From   (\ref{sigmap3}) it follows that all possible products with the terms of  $\tau_3$
with non-vanishing trace yield
\begin{equation}\label{Pid61}
 \mbox{Tr}\,\tau_3\,(\pi_0+\pi_3) = \frac{1}{\alpha'}
 \mbox{Tr} \left( |\varphi'\rangle\langle \varphi'| \varphi'\rangle\langle \varphi'| \right)
  = \frac{\alpha'^2}{\alpha'}=\alpha'
 \;.
\end{equation}

Similarly, we can evaluate $\Lambda_{00}+\Lambda_{03}= \mbox{Tr}\,\tau_0\,(\pi_0+\pi_3)$
using  (\ref{sigmap0}) and obtain
\begin{equation}\label{Pid62}
 \mbox{Tr}\,\tau_0\,(\pi_0+\pi_3) = \frac{1}{\alpha'}
 \mbox{Tr} \left( |\varphi'\rangle\langle \varphi'| \varphi'\rangle\langle \varphi'| \right)
  = \frac{\alpha'^2}{\alpha'}=\alpha'
 \;,
\end{equation}
which proves (\ref{Pid6}).

\section{Proof of Proposition \ref{CPP}}\label{ACPP}

The if part directly follows from Proposition \ref{PC}.
For the only-if part consider an arbitrary one-dimensional
projector $P_\varphi$ and the Kraus operator sum representation of  $T P_\varphi$:
\begin{equation}\label{CPP1}
T P_\varphi = \sum_{i=1}^{n} A_i\, | \varphi\rangle\langle\varphi| \,A_i^\ast
= \sum_{i=1}^{n}\, |  A_i\,\varphi\rangle\langle  A_i\,\varphi|
= \sum_{i=1}^{n}\, \left| \lambda_i \right|^2\, | \psi_i\rangle\langle \psi_i|
\;,
\end{equation}
where $A_i\,\phi = \lambda_i\;\psi_i$ such that $\| \psi_i \|=1$ and
$\lambda_i\in{\mathbb C}$ for all $i=1,\ldots,n$.

Since $T$ is pure, $T P_\varphi$ must be a one-dimensional projector, up to a non-negative factor.
Hence either $n=1$, which means that $T$ is of type (i), or $n>1$ and
all projectors $|\psi_i\rangle\langle \psi_i|$ must be identical to a constant projector
$P_\psi=|\psi\rangle\langle \psi|$. This proves the following Lemma:
%%%%%%%%%%%%%%%%%%%%%%%%%%%%%%%%%%%%%%%%%%%%%%%%%%%%%%%%%%%%%%%%%%%%%%%%%%%%%%%%%%%%%%%%%%%%%%%%%%
\begin{lemma}\label{LCPP}
 Let $T$ be a pure, completely positive operation with a Kraus operator sum
 representation (\ref{CPP1}) with a minimal number of terms $n>1$. Then for all $\varphi\in{\mathcal H}$ with $\| \varphi\|=1$
 there exists a $\psi\in{\mathcal H}$ with $\| \psi\|=1$ and complex numbers $\lambda_i\in{\mathbb C},\,i=1,\ldots,n,$ such that
 \begin{equation}\label{CPP1a}
  A_i\,\varphi = \lambda_i\;\psi\;\mbox{for all}\; i=1,\ldots,n\;.
\end{equation}
\end{lemma}
%%%%%%%%%%%%%%%%%%%%%%%%%%%%%%%%%%%%%%%%%%%%%%%%%%%%%%%%%%%%%%%%%%%%%%%%%%%%%%%%%%%%%%%%%%%%%%%%%%%%

For the remaining proof we hence will assume that $n>1$.
From Lemma (\ref{LCPP}) it follows that
\begin{equation}\label{CPP2}
T P_\varphi=\left( \sum_{i=1}^{n}\, \left| \lambda_i \right|^2\right)\,P_\psi
=\mbox{Tr } \left(W \,P_\varphi \right)\; P_\psi
\;,
\end{equation}
using that the trace of (\ref{CPP1}) reads:
\begin{equation}\label{CPP3}
\mbox{Tr }(T\,P_\varphi)=
  \sum_{i=1}^{n}\mbox{Tr }\left(A_i\,P_\varphi\,A_i^\ast\right) =
  \mbox{Tr }\underbrace{\left(\sum_{i=1}^{n}A_i^\ast\,A_i\right)}_{=:W\ge 0}\,P_\varphi
  =\mbox{Tr } \left(W \,P_\phi \right) =
  \sum_{i=1}^{n}\left| \lambda_i \right|^2
  \;.
\end{equation}

It remains to show that $P_\psi$  is one and the same projector for all $P_\varphi$, or, more precisely:
%%%%%%%%%%%%%%%%%%%%%%%%%%%%%%%%%%%%%%%%%%%%%%%%%%%%%%%%%%%%%%%%%%%%%%%%%%%%%%%%%%%%%%%%%%%%%%%%%%%%%%%%%%%%%%%%%%%
\begin{ass}\label{ASS1}
 There exists a  $\psi\in{\mathcal H},\; \|\psi\|=1$, such that for all $\varphi\in{\mathcal H},\; \|\varphi\|=1$
there exists some $\lambda\ge 0$ such that $T P_\varphi=\lambda\,P_{\psi}$.
\end{ass}
%%%%%%%%%%%%%%%%%%%%%%%%%%%%%%%%%%%%%%%%%%%%%%%%%%%%%%%%%%%%%%%%%%%%%%%%%%%%%%%%%%%%%%%%%%%%%%%%%%%%%%%%%%%%%%%%%%%%
Then, by linear extension, equation (\ref{CPP2}) holds for any $X\in B({\mathcal H})$, and thus $T$ would be of type (iii). \\

In order to prove Assertion \ref{ASS1} we consider the polar decomposition $ A_j=U_j\,P_j$ of the Kraus operators,
where $P_j$ is a positive semi-definite operator and $U_j$ a unitary one for $j=1,\ldots,n$.
Let $(\varphi_{ij})_{i=1,\ldots,N}$ denote the orthonormal eigenbasis of $P_j$ such that
\begin{equation}\label{eigenP}
  P_j\,\varphi_{ij}=\left\{
  \begin{array}{r@{\quad:\quad}l}
  p_{ij}\,\varphi_{ij} & i=1,\ldots,r_j\\
  0 &i=r_j+1,\ldots,N\;,
  \end{array}
  \right.
\end{equation}
and $p_{ij}>0$ for $i=1,\ldots,r_j,$ and $j=1,\ldots,n$.
$r_j$ denotes the rank of $A_j$.
It follows that
\begin{equation}\label{AeigenP}
A_j\,\varphi_{ij}=U_j\,P_j\,\varphi_{ij}= p_{ij}\,U_j \,\varphi_{ij}=:p_{ij}\,\psi_{ij}\;,\mbox{for }
 i=1,\ldots,r_j,\;\mbox{and } j=1,\ldots,n
 \;.
\end{equation}
Note that the $(\psi_{ij})_{i=1,\ldots,r_j}$ form an orthonormal basis of the subspace $A_j({\mathcal H})$
for all $j=1,\ldots,n$.

We will consider the following case distinction:
\begin{itemize}
  \item Case A: There exists a $1\le j\le n$ such that $r_j\ge 2$:
  At least one Kraus operator has rank $\ge 2$, or
  \item Case B: For all $1\le j\le n$ there holds $r_j=1$:
  All Kraus operators have rank $1$.
\end{itemize}

\subsection{Case A}\label{sec:CaseA}
Without loss of generality, we can assume $j=1$, since the numbering of the Kraus operators is arbitrary anyway.
Correspondingly, we will omit the index $j$ from $\varphi_{ij},\,\psi_{ij}$ and $p_{ij}$. Thus we obtain
\begin{equation}\label{CA1}
 P_1\,\varphi_1= p_1\,\varphi_1,\quad\mbox{and}\quad P_2\,\varphi_2= p_2\,\varphi_2\;.
\end{equation}
By applying Lemma (\ref{LCPP}) we conclude
\begin{equation}\label{CA2}
 A_k\,\varphi_1= \lambda_k\,\psi_1,\quad \mbox{and}\quad A_k\,\varphi_2= \mu_k\,\psi_2,\quad
 \mbox{for all}\,k=1,\ldots,n
 \;.
\end{equation}
Especially, $\lambda_1=p_1>0$ and $\mu_1=p_2>0$.

Let ${\mathcal H}_2$ denote the subspace of ${\mathcal H}$ spanned by $\varphi_1$ and $\varphi_2$
and, analogously, ${\mathcal H}'_2$  the subspace of ${\mathcal H}$ spanned by $\psi_1$ and $\psi_2$.
Define the vector
\begin{equation}\label{CA3}
  \Phi:= c_1 \,\varphi_1+c_2 \,\varphi_2, \quad c_i\in{\mathbb C},\; c_i\neq 0\;\mbox{for}\; i=1,2,\quad
  \mbox{such that } \| \Phi\|^2=|c_1|^2+|c_2|^2 =1
  \;.
\end{equation}
Again applying Lemma (\ref{LCPP}) we conclude
\begin{equation}
\label{CA4}
 A_k\,\Phi = \rho_k\,\Psi
   \stackrel{(\ref{CA3})}{=} c_1\,A_k\,\varphi_1 + c_2\,A_k\,\varphi_2\stackrel{(\ref{CA2})}{=}
    c_1\,\lambda_k\,\psi_1 + c_2\,\mu_k\,\psi_2
    \;,
\end{equation}
for all $k=1,\ldots,n$. Note that
\begin{equation}\label{CA5}
  \|\rho_1\,\Psi\|^2 =|c_1|^2\,p_1^2+|c_2|^2\,p_2^2 >0\quad \mbox{and hence}\; \rho_1\neq 0
  \;.
\end{equation}
By forming the scalar product of (\ref{CA4}) with $\psi_1$ and $\psi_2$ we conclude
\begin{equation}\label{CA6}
 \lambda_k=\frac{\langle \psi_1|\Psi\rangle }{c_1} \rho_k =: \nu_1\,\rho_k,\quad
 \mu_k=\frac{\langle \psi_2|\Psi\rangle }{c_2} \rho_k =: \nu_2\,\rho_k
 \;.
\end{equation}
Inserting this result into (\ref{CA2}) yields
\begin{eqnarray}
\label{CA7a}
  A_k\,\varphi_1 &=& \frac{\lambda_k}{\lambda_1} \lambda_1 \psi_1= \frac{\lambda_k}{\lambda_1} A_1\,\varphi_1
  =\frac{\nu_1}{\lambda_1}\rho_k\, A_1\,\varphi_1=:\kappa_1\,\rho_k\, A_1\,\varphi_1,\\
  \label{CA7b}
  A_k\,\varphi_2 &=& \frac{\mu_k}{\mu_1} \mu_1 \psi_2= \frac{\mu_k}{\mu_1} A_1\,\varphi_2
  =\frac{\nu_2}{\mu_1}\rho_k\, A_1\,\varphi_2=:\kappa_2\,\rho_k\, A_1\,\varphi_2\;.
\end{eqnarray}
We consider the basis $\left( |\varphi_i\rangle \langle\varphi_j|\right)_{1\le i,j\le 2}$
in $L({\mathcal H}_2)=B({\mathcal H}_2)\oplus {\sf i}B({\mathcal H}_2)$,
the space of all linear operators on ${\mathcal H}_2$ and
the extension of the operation $T$ to  $L({\mathcal H})=B({\mathcal H})\oplus {\sf i}B({\mathcal H})$.

It holds that
\begin{equation}\label{CA8}
 A_k\,|\varphi_i\rangle \langle \varphi_j |\,A_k^\ast \stackrel{(\ref{CA7a},\ref{CA7b})}{=}
 |\rho_k|^2 \, \kappa_i\,\overline{\kappa_j} A_1\,|\varphi_i\rangle \langle \varphi_j| \,A_1^\ast
 \;,
\end{equation}
for $1\le i,j\le 2$.
Hence
\begin{equation}\label{CA9}
 \sum_{k=1}^{n}A_k\,|\varphi_i\rangle \langle \varphi_j |\,A_k^\ast
 ={\mathcal A}\,|\varphi_i\rangle \langle \varphi_j |\,{\mathcal A}^\ast =:\widetilde{T}\left(|\varphi_i\rangle \langle \varphi_j | \right)
 \;,
\end{equation}
defining
\begin{equation}\label{CA10}
 {\mathcal A}:= \rho\,A_1\, \Delta,\quad \mbox{where} \
 \Delta= \kappa_1\,|\varphi_1\rangle\langle \varphi_1| +\kappa_2\,|\varphi_2\rangle\langle \varphi_2|
\quad \mbox{and} \quad
 \rho^2=\sum_{k=1}^{n}|\rho_k|^2
 \;.
\end{equation}
This means that the operation $T$ coincides with the pure operation $\widetilde{T}$, defined by
$\widetilde{T}\,X= {\mathcal A}\,X\,{\mathcal A}^\ast$ and hence being of type (i),
when restricted to $B({\mathcal H}_2)$. Let ${\mathcal F}'\in {\mathcal F}_3$ be the corresponding face of $K({\mathcal H})$
and $(\sigma_0({\mathcal F}),\sigma_1({\mathcal F}),\sigma_2({\mathcal F}),\sigma_3({\mathcal F}))$ be
a standard basis depending continuously on ${\mathcal F}\in {\mathcal F}_3$. Let $\det {\sf T}=\mbox{Det}
(T\,\sigma_0({\mathcal F}'),T\,\sigma_1({\mathcal F}'),T\,\sigma_2({\mathcal F}'),T\,\sigma_3({\mathcal F}'))$
and analogously for $\det \widetilde{\sf T}$.

In order to derive a contradiction let us assume that $T$ is of type (ii).
Then it follows that
$\det {\sf T}\le 0$, see the partial proof of Theorem \ref{Tpure} above. Analogously, $\det \widetilde{\sf T}\ge 0$,
since $\widetilde{T}$ is of type (i). Both determinants must coincide because
$T$  and $\widetilde{T}$ coincide on $B({\mathcal H}_2)$, which implies $\det {\sf T}=\det \widetilde{\sf T}=0$.
But this is impossible because $T$ maps $B({\mathcal H}_2)$ bijectively onto $B({\mathcal H}'_2)$.

$T$ cannot be of type (i) either since $n>1$.
This leaves only the possibility that $T$ is of type (iii),
which completes the proof for case A.

\subsection{Case B}\label{sec:CaseB}

In this case we will omit the index $i$ from $\varphi_{ij},\,\psi_{ij}$ and $p_{ij}$
and obtain
\begin{equation}\label{CB1}
  A_j\,\varphi_j = U_j\,P_j\,\varphi_j = p_j\, U_j\,\varphi_j = p_j\,\psi_j\quad \mbox{for}\quad j=1,\ldots,n.
\end{equation}
It follows that
\begin{equation}\label{CB2}
  A_j= p_j\,|\psi_j\rangle \langle \varphi_j| \quad \mbox{for}\quad j=1,\ldots,n
  \;,
\end{equation}
and hence for all projectors $P_\phi$ there holds
 \begin{equation}\label{CB3}
  T\,P_\phi = \sum_{j=1}^{n} A_j\,P_\phi \,A_j^\ast \stackrel{(\ref{CB2})}{=}
  \sum_{j=1}^{n} p_j^2 |\psi_j\rangle \langle \varphi_j|\phi\rangle \langle \phi|\varphi_j\rangle\langle \psi_j|
  = \sum_{j=1}^{n}  p_j^2  \left|\langle \varphi_j|\phi\rangle \right|^2\, P_{\psi_j}
  \;.
 \end{equation}
 Now we can argue analogously as in the proof of Lemma \ref{LCPP}:
 Since $T$ is a pure operation, $T\,P_\phi$ must be a non-negative multiple of a projection
 $P_\psi$ and hence all $P_{\psi_j},\,j=1,\ldots,n,$ must coincide with $P_\psi$.
 Moreover,
 \begin{equation}\label{CB4}
  T\,P_\phi  \stackrel{(\ref{CB3})}{=} \sum_{j=1}^{n}  p_j^2  \left|\langle \varphi_j|\phi\rangle \right|^2\, P_{\psi}
  =\mbox{Tr}\left(\underbrace{\left(\sum_{j=1}^{n}p_j^2\,P_{\varphi_j}\right)}_{=:W\ge 0} \,P_\phi \right) \,P_{\psi}
  =\mbox{Tr}(W\,P_\phi)\,P_\psi
  \;,
 \end{equation}
and hence $T$ will be of type (iii).
This completes the proof of Proposition  \ref{CPP}.

%%%%%%%%%%%%%%%%%%%%%%%%%%%%%%%%%%%%%%%%%%%%%%%%%%%%%%%%%%%%%%%%%%%%%%%%%%%%%%%%%%%%%%%%%%%%%%%%%%%%%%%%%%%%%%%%%%%%%%%%%%%%%%%%%%%%%%%%

\end{document}